\documentclass[envcountsame, runningheads]{llncs}

\usepackage{hyperref}

\usepackage{transparent}
\usepackage{amsmath,amssymb,amstext}

\usepackage{cite}
\usepackage{url}

\usepackage{enumitem}

\usepackage{graphicx}
\usepackage{wrapfig}

\usepackage{quiver}

\usepackage{tikz}
\usetikzlibrary{arrows,arrows.meta,automata,backgrounds,calc,chains,decorations.fractals,decorations.pathreplacing,fadings,fit,folding,mindmap,patterns,plotmarks,positioning,shadows,shapes.geometric,shapes.symbols,through,trees}
\usetikzlibrary{cd}

\colorlet{dblue}{blue!40!black}

\usepackage{enumitem}
\usepackage{verbatim}
\usepackage{float}

\usepackage{xspace}

\usepackage{latexsym}
\usepackage{fontawesome}

\usepackage{refcount}
\usepackage{caption}

\spnewtheorem{convention}[theorem]{Convention}{\bfseries}{\itshape}
\spnewtheorem{notation}[definition]{Notation}{\bfseries}{\itshape}

\newcounter{environment}

\newcommand{\store}[4]{
  \if\relax\detokenize{#2}\relax
  \expandafter\def\csname store#3\endcsname{\begin{#1}\label{#3}#4\end{#1}}%
  \else
  \expandafter\def\csname store#3\endcsname{\begin{#1}[#2]\label{#3}#4\end{#1}}%
  \fi
  \csname store#3\endcsname%
}

\newcommand{\use}[1]{{
  \renewcommand{\label}[1]{}%
  \renewcommand{\qedappendix}{}%
  \renewcommand{\footnote}[1]{}%
  \setcounter{environment}{\arabic{theorem}}%
  \setcounterref{theorem}{#1}%
  \addtocounter{theorem}{-1}%
  \csname store#1\endcsname%
  \setcounter{lemma}{\arabic{environment}}%
}}

\newcommand{\pbpostrong}{PBPO$^{+}$\xspace}
\newcommand{\catname}[1]{{\normalfont\textbf{#1}}\xspace}
\newcommand{\Set}{\catname{Set}}

\newcommand{\Graph}{\catname{Graph}}

\newcommand{\GraphLattice}[1]{\normalfont\textbf{Graph}$^{#1}$\xspace}

\newcommand{\labels}{\mathcal{L}}

\newcommand{\mono}{\rightarrowtail}

\newcommand{\pbpostep}[1]{\Rightarrow_{#1}}

\tikzset{default/.style={
  thick,
  every node/.style={circle},
  level distance=12mm, 
  inner sep=.5mm}}
\tikzset{smallCircle/.style={circle,fill=black,inner sep=0mm,outer sep=1mm,minimum size=1mm}}
\tikzset{paint/.style={very thick,draw=#1!50!black,fill=#1,opacity=.4}}
\tikzset{paintopaque/.style={very thick,draw=#1!50!black!60,fill=#1!60}}
\tikzset{loopl/.style={out=140,in=210,looseness=10}}
\tikzset{loopr/.style={out=-40,in=30,looseness=10}}
\tikzset{loopb/.style={out=-40-90,in=30-90,looseness=10}}
\tikzset{loopt/.style={out=-40+90,in=30+90,looseness=10}}
\tikzset{loop/.style={out=-30+#1,in=30+#1,looseness=10}}
\tikzset{thinloop/.style={out=-20+#1,in=20+#1,looseness=20}}

\tikzset{emptystep/.style={-,dotted,line cap=round,dash pattern=on 0 off 3.00000}}

\tikzset{b/.style={anchor=north,at=(#1.south)}}
\tikzset{br/.style={anchor=north west,at=(#1.south east)}}
\tikzset{bl/.style={anchor=north east,at=(#1.south west)}}
\tikzset{bw/.style={anchor=north west,at=(#1.south west)}}
\tikzset{be/.style={anchor=north east,at=(#1.south east)}}
\tikzset{a/.style={anchor=south,at=(#1.north)}}
\tikzset{ar/.style={anchor=south west,at=(#1.north east)}}
\tikzset{al/.style={anchor=south east,at=(#1.north west)}}
\tikzset{aw/.style={anchor=south west,at=(#1.north west)}}
\tikzset{ae/.style={anchor=south east,at=(#1.north east)}}
\tikzset{r/.style={anchor=west,at=(#1.east)}}
\tikzset{l/.style={anchor=east,at=(#1.west)}}
\tikzset{rn/.style={anchor=north west,at=(#1.north east)}}

\tikzset{eta/.style={very thick,->,cblue!90!black}}
\tikzset{beta/.style={very thick,->,corange!80!white!90!black}}

\tikzset{devcirc/.style={circle,draw,fill=white,inner sep=0,minimum size=4\pgflinewidth}}
\tikzset{dev/.style={postaction={decorate},decoration={
  markings,
  mark=at position .5 with \node [devcirc] {};}}
}
    
\tikzset{medium tree/.style={
    level 1/.style={sibling distance=17mm},
    level 2/.style={sibling distance=9mm},
    level 3/.style={sibling distance=5mm},
    level 4/.style={sibling distance=4mm},
  }}

\tikzset{startAt/.style={inner sep=0mm,r=#1,xshift=-1.2mm,yshift=.8mm}}

\newcommand{\arrowTriangle}[2]{
  \pgfpathmoveto{\pgfpoint{-0.01#1+#2}{.6#1}}
  \pgfpathlineto{\pgfpoint{#1+#2}{0}}
  \pgfpathlineto{\pgfpoint{-0.01#1+#2}{-.6#1}}
  \pgfusepathqfill
}

\newdimen\prearrowsize
\newdimen\arrowsize
\newdimen\temparrowsize
\newcommand{\arrowscale}{5}
\newcommand{\setarrowsize}{
  \arrowsize=0.000000001pt
  \prearrowsize=\arrowscale\pgflinewidth
  \normalizearrowsize
}
\newcommand{\normalizearrowsize}{
  \ifdim\prearrowsize>2mm
    \addtolength{\arrowsize}{2mm}
    \addtolength{\prearrowsize}{-2mm}
    \temparrowsize=0.5\prearrowsize
    \prearrowsize=\temparrowsize
  \else
  \fi

  \addtolength{\arrowsize}{\prearrowsize}
}

\pgfarrowsdeclarealias{share>}{share>}{stealth'}{stealth'}%

\pgfarrowsdeclare{red>}{red>}
{
  \setarrowsize
  \pgfarrowsleftextend{1.5\pgflinewidth} 
  \pgfarrowsrightextend{\arrowsize-0.1\arrowsize}
}{
  \setarrowsize
  \pgfsetdash{}{0pt} %
  \pgfsetroundjoin %
  \pgfsetroundcap %
  \arrowTriangle{\arrowsize}{-0.1\arrowsize}
}

\pgfarrowsdeclare{red>>}{red>>}
{
  \setarrowsize
  \pgfarrowsleftextend{1.5\pgflinewidth}
  \pgfarrowsrightextend{\arrowsize-0.1\arrowsize+.8\arrowsize}
}{
  \setarrowsize
  \pgfsetdash{}{0pt} %
  \pgfsetroundjoin %
  \pgfsetroundcap %
  \arrowTriangle{\arrowsize}{-0.1\arrowsize}
  \arrowTriangle{\arrowsize}{-0.1\arrowsize+.8\arrowsize}
}

\pgfarrowsdeclare{red>>>}{red>>>}
{
  \setarrowsize
  \pgfarrowsleftextend{1.5\pgflinewidth}
  \pgfarrowsrightextend{\arrowsize-0.1\arrowsize+1.6\arrowsize}
}{
  \setarrowsize
  \pgfsetdash{}{0pt} %
  \pgfsetroundjoin %
  \pgfsetroundcap %
  \arrowTriangle{\arrowsize}{-0.1\arrowsize}
  \arrowTriangle{\arrowsize}{-0.1\arrowsize+.8\arrowsize}
  \arrowTriangle{\arrowsize}{-0.1\arrowsize+1.6\arrowsize}
}

\tikzstyle{gyellow}=[draw=black!80,top color=white!50,bottom color=black!20]
\tikzstyle{gblue}=[draw=blue!50,top color=white,bottom color=blue!60]
\tikzstyle{gred}=[draw=red!50,top color=white,bottom color=red!60]
\tikzstyle{ggreen}=[draw=blue!80!green!90!black,top color=white,bottom color=blue!80!green!60]

\tikzstyle{roundNode}=[gyellow,thick,circle,minimum size=4mm,inner sep=0.5mm]

\definecolor{cblue}{rgb}{0,0.4,0.7}
\definecolor{clighterblue}{rgb}{0,0.6,1.0}
\colorlet{cred}{red}
\colorlet{cgreen}{green!80!black}
\colorlet{corange}{orange!70!red}
\colorlet{cpureorange}{orange}
\colorlet{cpurple}{clighterblue!50!cred}

\colorlet{clightblue}{clighterblue!50!cblue!40}%
\colorlet{clightred}{cred!40}
\colorlet{clightgreen}{cgreen!80!cblue!40}
\colorlet{clightyellow}{corange!40!yellow!50}
\colorlet{clightorange}{cred!50!orange!40}
\colorlet{clightpurple}{clighterblue!50!cred!50}

\colorlet{cdarkred}{cred!70!black}
\colorlet{cdarkgreen}{cgreen!60!black}
\colorlet{cdarkblue}{cblue!60!black}

\colorlet{chighlight}{orange!50!yellow!60}

\tikzset{pgnode/.style={smallCircle,fill=white,draw=black,minimum size=1.3mm,outer sep=0.5mm}}
\tikzset{pgnodecolor/.style={pgnode,minimum size=4mm,scale=0.9}}
\tikzset{pgnodebig/.style={roundNode,gyellow,outer sep=1mm}}
\tikzset{pgrelation/.style={ultra thick,cblue!80!black,decorate,decoration={snake,amplitude=.4mm,segment length=4mm}}}

\tikzset{exi/.style={densely dotted}}
\tikzset{decweak/.style={-,green!50!black,very thick,opacity=0.7}}
\tikzset{decstrict/.style={->,orange,very thick,opacity=0.7}}

\tikzset{graphNode/.style={circle,draw=black,inner sep=.5mm,outer sep=1mm}}

\tikzset{sloop/.style={looseness=7}}

\tikzset{interconnect/.style={dotted,cred}}
\tikzset{match/.style={cdarkgreen,very thick}}
\tikzset{jigsaw/.style={circle,draw=black,minimum size=2mm,fill=white,minimum size=3.5mm}}

\colorlet{myblue}{blue!80!black}
\colorlet{mygreen}{cdarkgreen}
\colorlet{myred}{cred}
\colorlet{myorange}{corange}
\colorlet{mypurple}{blue!40!cred}

\tikzset{smalljigsaw/.style={rectangle,rounded corners=3mm,inner sep=1mm}}

\tikzset{epattern/.style={}}
\tikzset{eset/.style={draw=black!50}}
\tikzset{npattern/.style={rectangle,rounded corners=2mm,draw=black,inner sep=0.5mm,outer sep=.5mm,minimum size=4.5mm}}
\tikzset{nset/.style={npattern,draw=black!50,fill=white}}
\tikzset{label/.style={scale=0.85,inner sep=0,outer sep=0.5mm}}

\tikzset{short/.style={node distance=10mm}}

\newcommand{\annotate}[2]{
  \node at (#1.north east) [anchor=south west,xshift=-.75mm,yshift=-.75mm,label,outer sep=0mm] {#2};
}

\newcommand{\graphbox}[8]{
  \begin{scope}[xshift=#2,yshift=#3]
    \draw [rounded corners=2mm] (0,0) rectangle (#4,-#5);
    \node at (0,0mm) [anchor=north west,inner sep=1mm] {#1};
    \begin{scope}[xshift=#4/2+#6,yshift=#7] 
    #8
    \end{scope}
  \end{scope}
}

\newcommand{\transparentgraphbox}[8]{
{
  \transparent{0.5}
  \graphbox{#1}{#2}{#3}{#4}{#5}{#6}{#7}{#8}
  }
}

\newcommand{\vertex}[2]{%
  \begin{tikzpicture}[baseline=-1ex]%
    \node [rectangle,rounded corners=2mm,inner sep=0.5mm,fill=#2] {$#1$};%
  \end{tikzpicture}%
}

\newcommand{\rulescale}{0.9}

\makeatletter

\newcommand*{\saveparinfo}[1]{%
  \expandafter\xdef\csname my@#1@parindent\endcsname{\the\parindent}%
  \expandafter\xdef\csname my@#1@parskip\endcsname{\the\parskip}%
}

\newcommand*{\useparinfo}[1]{%
  \parindent=\csname my@#1@parindent\endcsname \relax
  \parskip=\csname my@#1@parskip\endcsname \relax
}

\makeatother

\newcommand{\qedappendix}{\hfill \textcolor{blue}{$\circledast$}}

\newcommand{\PB}{{\normalfont PB}}
\newcommand{\PO}{{\normalfont PO}}

\usepackage{eso-pic}

\AddToShipoutPictureBG{
	\ifnum\value{page}=1 
	\AtPageLowerLeft{%
		\raisebox{3\baselineskip}{\makebox[\paperwidth]{\begin{minipage}{21cm}\centering
					The final authenticated version is available online at \url{https://doi.org/10.1007/978-3-030-78946-6_4}.
		\end{minipage}}}%
	}
	\fi
}

\tikzset{graphNode/.style={circle,draw=black,inner sep=.5mm,outer sep=1mm}}

\begin{document}

\title{Graph Rewriting and Relabeling with \pbpostrong}

\author{Roy Overbeek \and J\"{o}rg Endrullis \and Alo\"{i}s Rosset}
\authorrunning{Roy Overbeek \and J\"{o}rg Endrullis \and Alo\"{i}s Rosset}

\institute{Vrije Universiteit Amsterdam, Amsterdam, The Netherlands \\
  \email{\{r.overbeek, j.endrullis, a.rosset\}@vu.nl}
}
\maketitle

\begin{abstract}%
We extend the powerful Pullback-Pushout (PBPO) approach for graph rewriting with strong matching.
Our approach, called \pbpostrong, exerts more control over the embedding of the pattern in the host graph, which is important for a large class of graph rewrite systems.
In addition, we show that \pbpostrong is well-suited for rewriting labeled graphs and certain classes of attributed graphs.
For this purpose, we employ a lattice structure on the label set and use order-preserving graph morphisms.
We argue that our approach is simpler and more general than related relabeling approaches in the literature.
\end{abstract}

\saveparinfo{parinfo}%

\section{Introduction}
\label{sec:introduction}

Injectively matching a graph pattern $P$ into a host graph $G$ induces a classification of $G$ into three parts: (i)~a \emph{match graph} $M$, the image of $P$; (ii)~a \emph{context graph} $C$, the largest subgraph disjoint from $M$; and (iii)~a \emph{patch} $J$, the set of edges that are in neither $M$ nor $C$.
For example, if $P$ and $G$ are
respectively
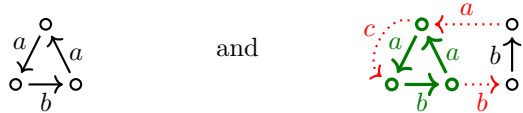
\begin{center}
    \begin{tikzpicture}[default,node distance=8mm,n/.style={graphNode}]
    \begin{scope}
    \node (3)[n] {};
    \node (4)[n] [right of=3] {};
    \node (5)[n] at ($(3)!.5!(4) + (0mm,8mm)$) {};
    \draw [->] (3) to node [below] {$b$} (4);
    \draw [->] (4) to node [right] {$a$} (5);
    \draw [->] (5) to node [above left] {$a$} (3);
    \end{scope}
    \node at (29.5mm,5mm) {and};
    \begin{scope}[xshift=50mm]
        \node (3)[n, match] {};
        \node (4)[n, match] [right of=3] {};
        \node (5)[n, match] at ($(3)!.5!(4) + (0mm,8mm)$) {};
        
        \node(6)[n] [right of=4] {};
        \node(7)[n] [above of=6] {};
        
        \draw [->, match] (3) to node [below] {$b$} (4);
        \draw [->, match] (4) to node [right] {$a$} (5);
        \draw [->, match] (5) to node [above left] {$a$} (3);
        \draw [->, myred, interconnect] (4) to node [below] {$b$} (6);
        \draw [->, myred, interconnect] (7) to node [above] {$a$} (5);
        \draw [->] (6) to node [left] {$b$} (7);
        \draw [->, myred, interconnect] (5) to[bend right=80] node [above left] {$c$} (3);
    \end{scope}
    \end{tikzpicture}
\end{center}
then $M$, $C$ and $J$ are indicated in green, black and red (and dotted), respectively.
We call this kind of classification a \emph{patch decomposition}.

Guided by the notion of patch decomposition, we recently introduced the expressive Patch Graph Rewriting (PGR) formalism~\cite{overbeek2020patch}. Like most graph rewriting formalisms, PGR rules specify a replacement of a left-hand side (lhs) pattern $L$ by a right-hand side (rhs) $R$. Unlike most rewriting formalisms, however, PGR rules allow one to (a)~constrain the permitted shapes of patches around a match for $L$, and (b)~specify how the permitted patches should be transformed, where transformations include rearrangement, deletion and duplication of patch edges. 

Whereas PGR is defined set-theoretically, in this paper we propose a more sophisticated categorical approach, inspired by the same ideas. Such an approach is valuable for at least three reasons: (i)~the classes of structures the method can be applied to is vastly generalized, (ii)~typical
meta-properties of interest (such as parallelism and concurrency) are more easily studied on the categorical level, and~(iii) it makes it easier to compare to existing categorical frameworks. 

The two main contributions of this paper are as follows.
First, we extend the Pullback Pushout (PBPO) approach by Corradini et al.~\cite{corradini2019pbpo} by strengthening the matching mechanism (Section~\ref{sec:pbpostrong}).
We call the resulting approach \emph{PBPO with strong matching}, or \emph{\pbpostrong} for short.
We argue that \pbpostrong is preferable over PBPO in situations where matching is nondeterministic, such as when specifying generative grammars or modeling execution.
Moreover, we show that in certain categories (including toposes), any PBPO rule can be modeled by a set of \pbpostrong rules (and even a single rule when matching is monic), while the converse does not hold (Section~\ref{sec:relating}).

Second, we show that \pbpostrong{} easily lends itself for rewriting labeled graphs and certain attributed graphs. To this end, we define a generalization of the usual category of labeled graphs, \GraphLattice{(\labels,\leq)}, in which the set of labels forms a complete lattice $(\labels,\leq)$ (Section~\ref{sec:category:graph:lattice}).
Not only does the combination of \pbpostrong and \GraphLattice{(\labels,\leq)} enable constraining and transforming the patch graph in flexible ways, it also provides natural support for modeling notions of relabeling, variables and sorts in rewrite rules. 
As we will clarify in the Discussion (Section~\ref{sec:discussion}), such mechanisms have typically been studied in the context of Double Pushout (DPO) rewriting~\cite{ehrig1973graph}, where the requirement to construct a pushout complement leads to technical complications and restrictions.

\section{Preliminaries}
\label{sec:preliminaries}

We assume familiarity with various basic categorical notions, notations and results, including morphisms $X \to Y$, pullbacks and pushouts, monomorphisms (monos) $X \mono Y$, identities $1_X : X \mono X$ and the pullback lemma~\cite{mac1971categories,awodey2006category}.

\newcommand{\src}{\mathit{s}}
\newcommand{\tgt}{\mathit{t}}
\newcommand{\lbl}{\ell}
\newcommand{\lblv}{\lbl^V}
\newcommand{\lble}{\lbl^E}

\begin{definition}[Graph Notions]
    A (labeled) \emph{graph} $G$ consists of a set of vertices $V$, a set of edges $E$, source and target functions $\src,\tgt : E \to V$, and label functions $\lblv : V \to \labels$ and $\lble : E \to \labels$ for some label set $\labels$.
    
    A graph is \emph{unlabeled} if $\labels$ is a singleton.
    
    A \emph{premorphism} between graphs $G$ and $G'$ is a pair of maps $\phi = (\phi_V : V_G \to V_{G'}, \phi_E : E_G \to E_{G'})$ satisfying
    $(s_{G'}, t_{G'}) \circ \phi_E = \phi_V \circ (s_G, t_G)$.
    
    A \emph{homomorphism} is a label-preserving premorphism $\phi$, i.e., a premorphism satisfying $\lblv_{G'} \circ \phi_V = \lblv_{G}$ and $\lble_{G'} \circ \phi_E = \lble_{G}$.
\end{definition}

\begin{definition}[Category $\Graph$~\cite{ehrig2006}]
    The category $\Graph$ has graphs as objects, parameterized over some global (and usually implicit) label set $\labels$, and homomorphisms as arrows. %
\end{definition}

Although we will point out similarities with PGR, an understanding of PGR is not required for understanding this paper. However, the following PGR terminology will prove useful (see also the opening paragraph of Section~\ref{sec:introduction}).

\begin{definition}[Patch Decomposition]
    Given a premorphism $x : X \to G$, we call the image $M = \mathit{im}(x)$ of $x$ the \emph{match graph} in $G$, $G - M$  the \emph{context graph} $C$ induced by $x$ (i.e., $C$ is the largest subgraph  disjoint from $M$), and the set of edges $E_G - E_M - E_C$ the set of \emph{patch edges} (or simply, \emph{patch}) induced by $x$. We refer to this decomposition induced by $x$ as a \emph{patch decomposition}. 
\end{definition}

\section{\pbpostrong}
\label{sec:pbpostrong}

We introduce \pbpostrong, which strengthens the matching mechanism of PBPO~\cite{corradini2019pbpo}. In the next section, we compare the two approaches and elaborate on the expressiveness of \pbpostrong.

\begin{definition}[\pbpostrong Rewrite Rule]\label{def:pbpostrong:rewrite:rule}
A \emph{\pbpostrong rewrite rule} $\rho$ is a collection of objects and morphisms, arranged as follows around a pullback square:
\begin{center}
  $\rho \ = \ $
  \begin{tikzpicture}[node distance=11mm,l/.style={inner sep=.5mm},baseline=-6mm]
    \node (L) {$L$};
    \node (K) [right of=L] {$K$}; \draw [->] (K) to node [above] {$l$} (L);
    \node (LP) [below of=L] {$L'$};
      \draw [>->] (L) to node [left] {$t_L$} (LP);
    \node (KP) [below of=K] {$K'$};
      \draw [>->] (K) to node [right] {$t_K$} (KP);
      \draw [->] (KP) to node [below] {$l'$} (LP);
       \node at ([shift={(-4mm,-4mm)}]K) {\PB};
    \node (R) [right of=K] {$R$}; 
      \draw [->] (K) to node [above] {$r$} (R);
  \end{tikzpicture}
\end{center}
$L$ is the \emph{lhs pattern} of the rule, $L'$ its \emph{type graph} and $t_L$ the \emph{typing} of $L$. Similarly for the \emph{interface} $K$. $R$ is the \emph{rhs pattern} or \emph{replacement for $L$}.
\end{definition}

\begin{remark}[A Mental Model for \Graph]
    In \Graph, $K'$ can be viewed as a collection of components, where every component is a (possibly generalized) subgraph of $L'$, as indicated by $l'$. By ``generalized'' we mean that the components may unfold loops and duplicate elements. $K$ is the restriction of $K'$ to those elements that are also in the image of $t_L$.
\end{remark}

We often depict the pushout $K' \xrightarrow{r'} R' \xleftarrow{t_R} R$ for span $K' \xleftarrow{t_K} K \xrightarrow{r} R$, because it shows the schematic effect of applying the rewrite rule. We reduce the opacity of $R'$ to emphasize that it is not part of the rule definition. 

\begin{example}[Rewrite Rule in \Graph]\label{ex:simple:rewrite:rule}
    A simple example of a rule for unlabeled graphs is the following:
    {%

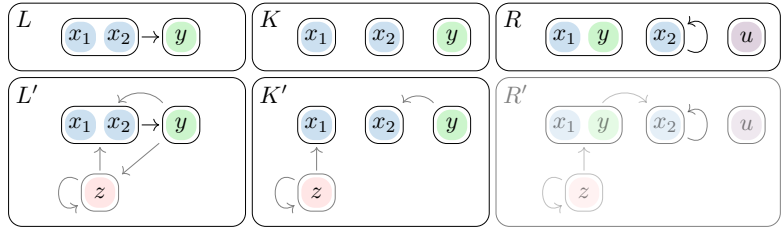
\newcommand{\nodexa}{\vertex{x_1}{cblue!20}}
\newcommand{\nodexb}{\vertex{x_2}{cblue!20}}
\newcommand{\nodey}{\vertex{y}{cgreen!20}}
\newcommand{\nodez}{\vertex{z}{cred!10}}
\newcommand{\nodeu}{\vertex{u}{cpurple!25}}

\begin{center}
  \scalebox{\rulescale}{
  \begin{tikzpicture}[->,node distance=12mm,n/.style={}]
    \graphbox{$L$}{0mm}{0mm}{35mm}{10mm}{-4mm}{-5mm}{
      \node [npattern] (xaxb)
      {\nodexa \ \nodexb};
      \node [npattern] (y) [right of=xaxb] {\nodey};
      \draw [epattern] (xaxb) to node {} (y);
    }
    \graphbox{$K$}{36mm}{0mm}{35mm}{10mm}{-8mm}{-5mm}{
      \node [npattern] (xa) {\nodexa};
      \node [npattern] (xb) [right of=xa, short] {\nodexb};
      \node [npattern] (y) [right of=xb, short] {\nodey};
    }
    \graphbox{$R$}{72mm}{0mm}{42mm}{10mm}{-8mm}{-5mm}{
      \node [npattern] (xay)
      {\nodexa \ \nodey};
      \node [npattern] (xb) [right of=xay] {\nodexb};
      \draw [epattern,loop=0,looseness=3] (xb) to node {} (xb);
      
      \node [npattern] (u) [right of=xb] {\nodeu};
      
    }
    \graphbox{$L'$}{0mm}{-11mm}{35mm}{22mm}{-4mm}{-7mm}{
      \node [npattern] (xaxb)
      {\nodexa \ \nodexb};
      \node [npattern] (y) [right of=xaxb] {\nodey};
      \node [nset] (z) [below of=xaxb, short] {\nodez};
      
      \draw [epattern] (xaxb) to node {} (y);
      
      \draw [eset] (y) to [bend right=50] node {} (xaxb);
      \draw [eset] (z) to node {} (xaxb);
      \draw [eset,loop=180,looseness=3] (z) to node {} (z);
      \draw [eset] (y) to node {} (z);
    }
    \graphbox{$K'$}{36mm}{-11mm}{35mm}{22mm}{-8mm}{-7mm}{
      \node [npattern] (xa) {\nodexa};
      \node [npattern] (xb) [right of=xa, short] {\nodexb};
      \node [npattern] (y) [right of=xb, short] {\nodey};
      \node [nset] (z) [below of=xa, short] {\nodez};
      
      \draw [eset] (y) to [bend right=50] node {} (xb);
      \draw [eset,loop=180,looseness=3] (z) to node {} (z);
      
      \draw [eset] (z) to node {} (xa);
    }
    \transparentgraphbox{$R'$}{72mm}{-11mm}{42mm}{22mm}{-8mm}{-7mm}{
      \node [npattern] (xay)
      {\nodexa \ \nodey};
      \node [npattern] (xb) [right of=xay] {\nodexb};
      \draw [epattern,loop=0,looseness=3] (xb) to node {} (xb);
      \node [nset] (z) [below of=xay, short] {\nodez};
      
      \draw [eset] (xay) to [bend left=50] node {} (xb);
      \draw [eset,loop=180,looseness=3] (z) to node {} (z);
      
      \draw [eset] (z) to node {} (xay);
      \node [npattern] (u) [right of=xb] {\nodeu};
    }
  \end{tikzpicture}
  }
\end{center}

}%
    \noindent
    In this and subsequent examples, a vertex is a  non-empty set $\{ x_1,\ldots,x_n \}$ represented by a box
    {
    \hspace{-3.4mm}
    \newcommand{\nodexa}{\vertex{x_1}{cblue!20}}
    \newcommand{\nodexb}{\vertex{x_n}{cblue!20}}
    \begin{tikzcd}[->,node distance=12mm,n/.style={}]
    \node [npattern] (xaxb)
          {\nodexa \ \raisebox{1mm}{$\cdots$} \  \nodexb};
    \end{tikzcd}
    }, and each morphism $f = (\phi_V, \phi_E) : G \to G'$ is the unique morphism satisfying $S \subseteq f(S)$ for all $S \in V_G$. For instance, for $\{ x_1 \}, \{ x_2 \} \in V_K$,  $l(\{x_1\}) = l(\{ x_2 \}) = \{ x_1,x_2 \} \in V_L$. We will use examples that ensure uniqueness of 
    each $f$ (in particular, we ensure that $\phi_E$ is uniquely determined). Colors are purely supplementary.
\end{example}

\newcommand{\picturematch}{
  \begin{tikzpicture}[node distance=13mm,l/.style={inner sep=1mm},baseline=-7.5mm,v/.style={node distance=11.5mm}]
    \node (G) {$G_L$};
    \node (L) [left of=G] {$L$};
      \draw [->] (L) to node [above,l] {$m$} (G);
    \node (L2) [below of=L,v] {$L$};
      \draw [>->] (L) to node [left,l] {$1_L$} (L2);
      \node at ([shift={(4mm,-4mm)}]L) {PB};
    \node (LP) [below of=G,v] {$L'$};
      \draw [->] (L2) to node [above,l] {$t_L$} (LP);
      \draw [->] (G) to node [left,l] {$\alpha$} (LP);
  \end{tikzpicture}
}

\newcommand{\strongmatch}[3]{\mathrm{strong}(#1,#2,#3)}

\noindent%
\begin{minipage}[c]{.79\textwidth}%
\begin{definition}[Strong Match]\label{def:pbpostrong:match}
A \emph{match morphism} $m$ and an \emph{adherence morphism} $\alpha$
 form a strong match for a typing $t_L$, denoted $\strongmatch{t_L}{m}{\alpha}$, if the square on the right is a pullback square.
\end{definition}
\end{minipage}\hfill%
\begin{minipage}[c]{.18\textwidth}
    \hfill
    \picturematch{}%
\end{minipage}

\begin{remark}[Preimage Interpretation]
    \label{remark:preimage:interpretation}
    In \Set-like categories (such as \Graph), the match diagram states that the preimage of $t_L(L)$ under $\alpha : G_L \to L'$ is $L$ itself. So each element of $t_L(L)$ is the $\alpha$-image of exactly one element of $G_L$.
\end{remark}

In practice it is natural to first fix a match $m$, and to subsequently verify whether it can be extended into a suitable adherence morphism $\alpha$.

\newcommand{\picturerule}{
  \begin{tikzpicture}[node distance=11mm,l/.style={inner sep=.5mm},baseline=-6mm]
    \node (L) {$L$};
    \node (K) [right of=L] {$K$}; \draw [->] (K) to node [above] {$l$} (L);
    \node (LP) [below of=L] {$L'$};
      \draw [>->] (L) to node [left] {$t_L$} (LP);
    \node (KP) [below of=K] {$K'$};
      \draw [>->] (K) to node [right] {$t_K$} (KP);
      \draw [->] (KP) to node [below] {$l'$} (LP);
       \node at ([shift={(-4mm,-4mm)}]K) {\PB};
    \node (R) [right of=K] {$R$}; 
      \draw [->] (K) to node [above] {$r$} (R);
  \end{tikzpicture}
}

\newcommand{\picturestep}{
  \begin{tikzpicture}[node distance=16mm,l/.style={inner sep=1mm},baseline=-7.5mm,v/.style={node distance=11mm}]
    \node (G) {$G_L$};
    \node (L) [left of=G] {$L$};
      \draw [>->] (L) to node [above,l] {$m$} (G);
    \node (L2) [below of=L,v] {$L$};
      \draw [>->] (L) to node [left,l] {$1_L$} (L2);
      \node at ([shift={(4mm,-4mm)}]L) {\PB};
    \node (LP) [below of=G,v] {$L'$};
      \draw [>->] (L2) to node [below,l] {$t_L$} (LP);
      \draw [->] (G) to node [left,l] {$\alpha$} (LP);
     \node (GKP) [right of=G] {$G_K$};
     \draw [->] (GKP) to node [above,l] {$g_L$} (G);
     \node (KP) [below of=GKP,v] {$K'$};
     \draw [->] (GKP) to node [right,l] {$u'$} (KP);
     \draw [->] (KP) to node [below,l] {$l'$} (LP);
     \node at ([shift={(-4mm,-4mm)}]GKP) {\PB};
     \node (K) [above of=GKP,v] {$K$};
     \draw [>->, dotted] (K) to node [left,l] {$!u$} (GKP);
     \node (R) [right of=K] {$R$};
     \draw [->] (K) to node [above,l] {$r$} (R);
     \node (GP) [below of=R,v] {$G_R$};
     \draw [->] (GKP) to node [below,l,pos=0.7] {$g_R$} (GP);
     \draw [->] (R) to node [right,l] {$w$} (GP);
     \node at ([shift={(-4mm,4mm)}]GP) {\PO};

     \draw [>->] (K) to[out=-45,in=30,looseness=0.8] node[right,l,pos=0.85] {$t_K$} (KP);
  \end{tikzpicture}
}

\begin{definition}[\pbpostrong Rewrite Step] \label{def:pbpostrong:rewrite:step}
    A \pbpostrong rewrite rule $\rho$ (left)
    and adherence morphism $\alpha : G_L \to L'$ induce a rewrite step $G_L \Rightarrow_\rho^\alpha G_R$ on arbitrary $G_L$ and $G_R$ if the properties indicated by the commuting diagram (right)
    \begin{center}
      \raisebox{9mm}{$\rho \ = \ $\picturerule} \hspace{1cm} \picturestep
    \end{center}
    hold, where $u : K \to G_K$ is the unique (and necessarily monic) morphism satisfying $t_K = u' \circ u$. We write $G_L \Rightarrow_\rho G_R$ if $G_L \Rightarrow_\rho^\alpha G_R$ for some $\alpha$.
\end{definition}

It can be seen that the rewrite step diagram consists of a match square, a pullback square for extracting (and possibly duplicating) parts of $G_L$, and finally a pushout square for gluing these parts along pattern $R$.

The following lemma establishes the existence of a monic $u$ by constructing a witness, and Lemma~\ref{lem:uniqueness:u} establishes uniqueness.

\store{lemma}{Top-Left Pullback}{lem:topleft:pullback}{
    In the rewrite step diagram of Definition~\ref{def:pbpostrong:rewrite:step}, there exists a morphism $u : K \to G_K$ such that
    $L \xleftarrow{l} K \xrightarrow{u} G_K$ is a pullback for $L \xrightarrow{m} G_L \xleftarrow{g_L} G_K$,
    $t_K = u' \circ u$, and
    $u$ is monic. \qedappendix%
    \footnote{We use $\textcolor{blue}{\circledast}$ instead of $\qed$ when the proof is available in the Appendix.}
}

\store{lemma}{Uniqueness of $u$}{lem:uniqueness:u}{
    In the rewrite step diagram of Definition~\ref{def:pbpostrong:rewrite:step} (and in any category), there is a unique $v : K \to G_K$ such that $t_K = u' \circ v$. \qedappendix
}

\store{lemma}{Bottom-Right Pushout}{lem:bottomright:pushout}{
    Let $K' \xrightarrow{r'} R' \xleftarrow{t_R} R$ be a pushout for cospan $R \xleftarrow{r} K \xrightarrow{t_K} K'$ of rule $\rho$ in Definition~\ref{def:pbpostrong:rewrite:step}. Then in the rewrite step diagram, there exists a morphism $w' : G_R \to R'$ such that
    $t_R = w' \circ w$, and
    $K' \xrightarrow{r'} R' \xleftarrow{w'} G_R$ is a pushout for $K' \xleftarrow{u'} G_K \xrightarrow{g_R} G_R$. \qedappendix
}

Lemmas~\ref{lem:topleft:pullback} and \ref{lem:bottomright:pushout} show that a \pbpostrong step defines a commuting diagram similar to the PBPO definition (Definition~\ref{def:pbpo:rewrite:step}):
\begin{center}
\vspace{-1ex}
\begin{tikzpicture}[->,l/.style={label,inner sep=1mm},erase/.style={-,line width=1mm,white}]
    \begin{scope}[node distance=3.2cm]
      \node[opacity=0.4] (L) {$L$};
      \node (K) [right of=L] {$K$};
      \node (R) [right of=K] {$R$};
      \draw[opacity=0.4] (K) to node [above,l] {$l$} (L);
      \draw (K) to node [above,l] {$r$} (R);
    \end{scope}
    \begin{scope}[node distance=.8cm]
      \node (GL) [below of=L] {$G_L$};
      \node (GK) [below of=K] {$G_K$};
      \node (GR) [below of=R] {$G_R$};
      \draw (GK) to node [l,fill=white] {$g_L$} (GL);
      \draw (GK) to node [l,fill=white] {$g_R$} (GR);
      \draw [>->, opacity=0.5] (L) to node [right,l] {$m$} (GL);
      \draw [>->, densely dotted] (K) to node [right,l] {$!u$} (GK);
      \draw (R) to node [right,l] {$w$} (GR);
    \end{scope}
    \begin{scope}[node distance=.8cm]
      \node (L') [below of=GL] {$L'$};
      \node (K') [below of=GK] {$K'$};
      \node[opacity=0.4] (R') [below of=GR] {$R'$};
      \draw (K') to node [below,l] {$l'$} (L');
      \draw[opacity=0.4] (K') to node [below,l] {$r'$} (R');
      \draw (GL) to node [right,l] {$\alpha$} (L');
      \draw (GK) to node [right,l] {$u'$} (K');
      \draw[opacity=0.5] (GR) to node [right,l] {$w'$} (R');
    \end{scope}
    \begin{scope}[twist/.style={bend left=70}]
      \draw [erase](L) to[twist] (L');
      \draw [>->, opacity=0.5] (L) to[twist] node [right,pos=0.7] {$t_L$} (L');
      \draw [erase] (K) to[twist] (K');
      \draw [>->] (K) to[twist] node [right,pos=0.7] {$t_K$} (K');
      \draw [erase] (R) to[twist] (R');
      \draw [opacity=0.4] (R) to[twist] node [right,pos=0.7] {$t_R$} (R');
    \end{scope}
    \node[opacity=0.4] at ($(L)!.5!(GK)$) {PB};
    \node at ($(L')!.5!(GK)$) {PB};
    \node at ($(R)!.5!(GK)$) {PO};
    \node[opacity=0.4] at ($(R')!.5!(GK)$) {PO};
    \begin{scope}[node distance=1.2cm]
      \node (L1) [left of=GL] {$L$};
      \node (L2) [left of=L'] {$L$};
      \draw [>->] (L1) to node [l,above] {$m$} (GL);
      \draw [>->] (L1) to node [l,left] {$1_L$} (L2);
      \draw [>->] (L2) to node [l,below] {$t_L$} (L');
    \node at ($(L1)!.5!(L')$) {PB};
    \end{scope}
  \end{tikzpicture}
\vspace{-1ex}
\end{center}
We will omit the match diagram in depictions of steps.
\pagebreak

\begin{example}[Rewrite Step]
\label{ex:simple:rewrite:step}
Applying the rule given in Example~\ref{ex:simple:rewrite:rule} to $G_L$ (as depicted below) has the following effect:
{%

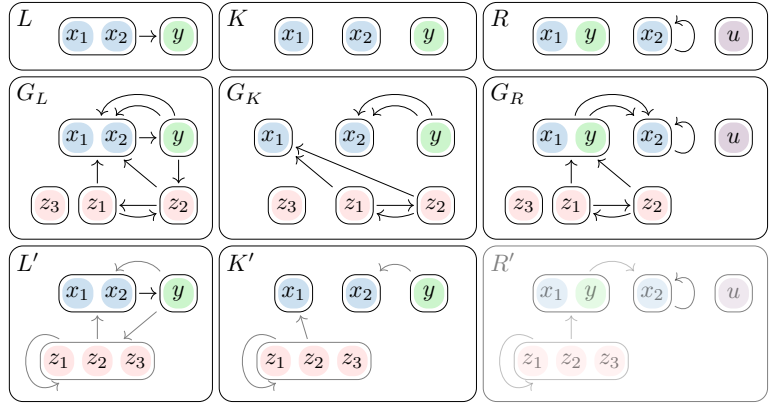
\newcommand{\nodexa}{\vertex{x_1}{cblue!20}}%
\newcommand{\nodexb}{\vertex{x_2}{cblue!20}}%
\newcommand{\nodey}{\vertex{y}{cgreen!20}}%
\newcommand{\nodez}{\vertex{z}{cred!10}}%
\newcommand{\nodeza}{\vertex{z_1}{cred!10}}%
\newcommand{\nodezb}{\vertex{z_2}{cred!10}}%
\newcommand{\nodezc}{\vertex{z_3}{cred!10}}%
\newcommand{\nodeu}{\vertex{u}{cpurple!25}}%
\begin{center}\vspace{-.25ex}
  \scalebox{\rulescale}{
  \begin{tikzpicture}[->,node distance=12mm,n/.style={}]
      \graphbox{$L$}{0mm}{0mm}{30mm}{10mm}{-2mm}{-5mm}{
      \node [npattern] (xaxb)
      {\nodexa \ \nodexb};
      \node [npattern] (y) [right of=xaxb] {\nodey};
      \draw [epattern] (xaxb) to node {} (y);
    }
    \graphbox{$G_L$}{0mm}{-11mm}{30mm}{24mm}{-2mm}{-9mm}{
      \node [npattern] (xaxb)
      {\nodexa \ \nodexb};
      \node [npattern] (y) [right of=xaxb] {\nodey};
      \draw [epattern] (xaxb) to node {} (y);
      
      \node [npattern] (za) [below of=xaxb, short] {\nodeza};
      \node [npattern] (zb) [right of=za] {\nodezb};
      \node [npattern] (zc) [left of=za,node distance=7mm] {\nodezc};
      
      \draw [epattern] (za) to node {} (xaxb);
      \draw [epattern] (zb) to node {} (za);
      \draw [epattern] (zb) to node {} (xaxb);
      \draw [epattern] (za) to [bend right=20] node {} (zb);
      
      \draw [epattern] (y) to [bend right=50] node {} (xaxb);
      \draw [epattern] (y) to [bend right=80] node {} (xaxb);
      \draw [epattern] (y) to (zb);
    }
    \graphbox{$G_K$}{31mm}{-11mm}{38mm}{24mm}{-2mm}{-9mm}{
      \node [npattern, xshift=3mm] (xaxb)
      {\nodexb};
      \node [npattern] (xa) [left of=xaxb]
      {\nodexa};
      \node [npattern] (y) [right of=xaxb] {\nodey};
      
      \node [npattern] (za) [below of=xaxb, short] {\nodeza};
      \node [npattern] (zb) [right of=za] {\nodezb};
      \node [npattern] (zc) [left of=za,node distance=10mm] {\nodezc};
      
      \draw [epattern] (zb) to [bend left=20] node {} (za);
      \draw [epattern] (za) to node {} (zb);
      
      \draw [epattern] (y) to [bend right=50] node {} (xaxb);
      \draw [epattern] (y) to [bend right=80] node {} (xaxb);
      \draw [epattern] (za) to node {} (xa);
      \draw [epattern] (zb) to node {} (xa);
      
    }
    \graphbox{$K$}{31mm}{0mm}{38mm}{10mm}{-8mm}{-5mm}{
      \node [npattern] (xa) {\nodexa};
      \node [npattern] (xb) [right of=xa, short] {\nodexb};
      \node [npattern] (y) [right of=xb, short] {\nodey};
    }
    
    \graphbox{$R$}{70mm}{0mm}{42mm}{10mm}{-8mm}{-5mm}{
      \node [npattern] (xay)
      {\nodexa \ \nodey};
      \node [npattern] (xb) [right of=xay] {\nodexb};
      \draw [epattern,loop=0,looseness=3] (xb) to node {} (xb);
      
      \node [npattern] (u) [right of=xb] {\nodeu};
    }
    
    \graphbox{$G_R$}{70mm}{-11mm}{42mm}{24mm}{-8mm}{-9mm}{
      \node [npattern] (xaxb)
      {\nodexa \ \nodey};
      \node [npattern] (y) [right of=xaxb] {\nodexb};
      
      \node [npattern] (za) [below of=xaxb, short] {\nodeza};
      \node [npattern] (zb) [right of=za] {\nodezb};
      \node [npattern] (zc) [left of=za,node distance=7mm] {\nodezc};

      \node [npattern] (u) [right of=y] {\nodeu};
      
      \draw [epattern] (zb) to [bend left=20] node {} (za);
      \draw [epattern] (za) to node {} (zb);
      
      \draw [epattern] (xaxb) to [bend left=50] node {} (y);
      \draw [epattern] (xaxb) to [bend left=80] node {} (y);
      \draw [epattern,loop=0,looseness=3] (y) to node {} (y);
      
      \draw [epattern] (za) to node {} (xaxb);
      \draw [epattern] (zb) to node {} (xaxb);
    }
    \graphbox{$L'$}{0mm}{-36mm}{30mm}{23mm}{-2mm}{-7mm}{
      \node [npattern] (xaxb)
      {\nodexa \ \nodexb};
      \node [npattern] (y) [right of=xaxb] {\nodey};
      \node [nset] (z) [below of=xaxb, short]
      {\nodeza \ \nodezb \ \nodezc};
      
      \draw [epattern] (xaxb) to node {} (y);
      
      \draw [eset] (y) to [bend right=50] node {} (xaxb);
      \draw [eset] (z) to node {} (xaxb);
      \draw [eset,loop=180,looseness=3] (z) to node {} (z);
      
      \draw[eset] (y) to node {} (z);
    }
    \graphbox{$K'$}{31mm}{-36mm}{38mm}{23mm}{-8mm}{-7mm}{
      \node [npattern] (xa) {\nodexa };
      \node [npattern] (xb) [right of=xa, short] {\nodexb};
      \node [npattern] (y) [right of=xb, short] {\nodey};
      \node [nset] (z) [below of=xa,short, xshift=3mm] {\nodeza \ \nodezb \ \nodezc};
      
      \draw [eset] (y) to [bend right=50] node {} (xb);
      \draw [eset,loop=180,looseness=3] (z) to node {} (z);
      \draw [eset] (z) to node {} (xa);
    }
    \transparentgraphbox{$R'$}{70mm}{-36mm}{42mm}{23mm}{-8mm}{-7mm}{
      \node [npattern] (xay)
      {\nodexa \ \nodey};
      \node [npattern] (xb) [right of=xay] {\nodexb};
      \draw [epattern,loop=0,looseness=3] (xb) to node {} (xb);
      \node [nset] (z) [below of=xay, short] {\nodeza \ \nodezb \ \nodezc};
      
      \draw [eset] (xay) to [bend left=50] node {} (xb);
      \draw [eset,loop=180,looseness=3] (z) to node {} (z);
      
      \draw [eset] (z) to node {} (xay);
      \node [npattern] (u) [right of=xb] {\nodeu};
    }
  \end{tikzpicture}
  }\vspace{-1ex}
\end{center}

}%
This example illustrates (i)~how permitted patches can be constrained (e.g., $L'$ forbids patch edges targeting $y$), (ii)~how patch edge endpoints that lie in the image of $t_L$ can be redefined, and (iii)~how patch edges can be deleted.
\end{example}

In the examples of this section, we have restricted our attention to unlabeled graphs. In Section~\ref{sec:category:graph:lattice}, we show that the category \GraphLattice{(\labels,\leq)} is more suitable than \Graph for rewriting labeled graphs using \pbpostrong{}.

\section{Expressiveness of \pbpostrong}
\label{sec:relating}

The set of \pbpostrong rules is a strict subset of the set of PBPO rules, and for any \pbpostrong rule $\rho$, we have ${\Rightarrow_\rho^{\mathrm{PBPO}^{+}}} \subseteq {\Rightarrow_\rho^{\mathrm{PBPO}}}$ for the generated rewrite relations. Nevertheless, we will show that under certain assumptions, any PBPO rule can be modeled by a set of \pbpostrong rules, but not vice versa. Thus, in many categories of interest (such as toposes), \pbpostrong can define strictly more expressive grammars than PBPO.
This result may be likened to Habel et al.'s result that restricting DPO to monic matching increases expressive power~\cite{habel2001doublerevisited}.

In Section~\ref{sec:relating:rule:match:step}, we recall and compare the PBPO definitions for rule, match and step, clarifying why \pbpostrong is shorthand for \emph{PBPO with strong matching}. We then argue why strong matching is usually desirable in Section~\ref{sec:relating:case:for:strong:matching}. Finally, we prove a number of novel results on PBPO in Section~\ref{sec:relating:simulating}, relating to monic matching, monic rules and strong matching. Our claim about \pbpostrong's expressiveness follows as a consequence.

\subsection{PBPO: Rule, Match \& Step}
\label{sec:relating:rule:match:step}

\smallskip

\noindent
\begin{minipage}{0.68\textwidth}
    \begin{definition}[PBPO Rule~\cite{corradini2019pbpo}]
        \label{def:pbpo:rule}
        A \emph{PBPO rule} $\rho$ is a commutative diagram as shown on the right.
        The bottom span can be regarded as a typing for the top span. The rule is in \emph{canonical form} if the left square is a pullback and the right square is a pushout.
    \end{definition}
\end{minipage}\hfill%
\begin{minipage}{0.25\textwidth}
    \begin{tikzcd}[column sep=5mm,row sep=5mm]
    L \arrow[d, "t_L"']         & K \arrow[l, "l"'] \arrow[r, "r"] \arrow[d, "t_K"]            & R \arrow[d, "t_R"] \\
    L' \arrow[ru, "=", phantom] & K' \arrow[l, "l'"] \arrow[r, "r'"'] \arrow[ru, "=", phantom] & R'                
    \end{tikzcd}
\end{minipage}
\medskip

Every PBPO rule is equivalent to a rule in canonical form~\cite{corradini2019pbpo}, and in \pbpostrong, rules are limited to those in canonical form. The only important difference between a canonical PBPO rule and a \pbpostrong rule, then, is that a \pbpostrong rule requires monicity of $t_L$ (and hence also of $t_K$).

\begin{definition}[PBPO Match~\cite{corradini2019pbpo}]
A \emph{PBPO match} for a typing $t_L : L \to L'$ is a pair of morphisms $(m : L \to G,\alpha : G \to L')$ such that $t_L = \alpha \circ m$.
\end{definition}

The pullback construction used to establish a match in \pbpostrong implies $t_L = \alpha \circ m$. Thus PBPO matches are more general than the strong match used in \pbpostrong (Definition~\ref{def:pbpostrong:match}). More specifically for \Graph, PBPO allows mapping elements of the host graph $G_L$ not in the image of $m : L \to G_L$ onto the image of $t_L$, whereas \pbpostrong forbids this. In the next subsection, we will argue why it is often desirable to forbid such mappings.

\bigskip

\noindent
\begin{minipage}{0.57\textwidth}
    \begin{definition}[PBPO Rewrite Step~\cite{corradini2019pbpo}]\label{def:pbpo:rewrite:step}
        A PBPO rule $\rho$ (as in Definition~\ref{def:pbpo:rule}) induces a \emph{PBPO step} $G_L \Rightarrow_\rho^{m, \alpha} G_R$
        shown on the right, 
        where (i)~$u : K \to G_K$ is uniquely determined by the universal property of pullbacks and makes the top-left square commuting, 
        (ii)~$w' : G_R \to R'$\parfillskip=0pt
    \end{definition}
\end{minipage}\hfill%
\begin{minipage}{0.42\textwidth}
      \begin{tikzpicture}[->,l/.style={label,inner sep=1mm}]
        \begin{scope}[node distance=2cm]
          \node (L) {$L$};
          \node (K) [right of=L] {$K$};
          \node (R) [right of=K] {$R$};
          \draw (K) to node [above,l] {$l$} (L);
          \draw (K) to node [above,l] {$r$} (R);
        \end{scope}
        \begin{scope}[node distance=0.9cm]
          \node (GL) [below of=L] {$G_L$};
          \node (GK) [below of=K] {$G_K$};
          \node (GR) [below of=R] {$G_R$};
          \draw (GK) to node [l,fill=white] {$g_L$} (GL);
          \draw (GK) to node [l,fill=white] {$g_R$} (GR);
          \draw (L) to node [right,l] {$m$} (GL);
          \draw (K) to node [right,l] {$u$} (GK);
          \draw (R) to node [right,l] {$w$} (GR);
        \end{scope}
        \begin{scope}[node distance=0.9cm]
          \node (L') [below of=GL] {$L'$};
          \node (K') [below of=GK] {$K'$};
          \node (R') [below of=GR] {$R'$};
          \draw (K') to node [below,l] {$l'$} (L');
          \draw (K') to node [below,l] {$r'$} (R');
          \draw (GL) to node [right,l] {$\alpha$} (L');
          \draw (GK) to node [right,l] {$u'$} (K');
          \draw (GR) to node [right,l] {$w'$} (R');
        \end{scope}
        \begin{scope}[twist/.style={bend right=70},erase/.style={-,line width=1mm,white}]
          \draw (L) to[twist] node [right,pos=0.7] {$t_L$} (L');
          \draw [erase] (K) to[twist] (K');
          \draw (K) to[twist] node [right,pos=0.7] {$t_K$} (K');
          \draw [erase] (R) to[twist] (R');
          \draw (R) to[twist] node [right,pos=0.7] {$t_R$} (R');
        \end{scope}
        \node at ($(L)!.5!(GK) + (-2mm,0mm)$) {$=$};
        \node at ($(L')!.5!(GK) + (-2mm,0mm)$) {PB};
        \node at ($(R)!.5!(GK) + (-2mm,0mm)$) {PO};
        \node at ($(R')!.5!(GK) + (-2mm,0mm)$) {$=$};
      \end{tikzpicture}\vspace{2.2ex}
\end{minipage}\\
{\itshape is uniquely determined by the universal property of pushouts and makes the bottom-right square commuting, and
$t_L = \alpha \circ m$.

We write 
$G_L \pbpostep{\rho} G_R$ if
$G_L \Rightarrow_\rho^{m, \alpha} G_R$ for some $m$ and $\alpha$.
}
\medskip

The match square of \pbpostrong allows simplifying the characterization of $u$, as shown in the proof to Lemma~\ref{lem:uniqueness:u}. This simplification is not possible for PBPO (see Remark~\ref{remark:pbpo:u:not:unique}). The bottom-right square is omitted in the definition of a \pbpostrong rewrite step, but can be reconstructed through a pushout (modulo isomorphism). So this difference is not essential.

\begin{figure}[b!]
    \centering
    {%

\newcommand{\nodexa}{\vertex{x_1}{cblue!20}}
\newcommand{\nodexb}{\vertex{x_2}{cblue!20}}
\newcommand{\nodexc}{\vertex{x_3}{cblue!20}}
\newcommand{\nodexd}{\vertex{x_4}{cblue!20}}

\newcommand{\nodea}{\vertex{a}{cgreen!20}}
\newcommand{\nodeb}{\vertex{b}{cpurple!25}}

\newcommand{\nodeaa}{\vertex{a_1}{cgreen!20}}
\newcommand{\nodeab}{\vertex{a_2}{cgreen!20}}
\newcommand{\nodeba}{\vertex{b_1}{cpurple!25}}
\newcommand{\nodebb}{\vertex{b_2}{cpurple!25}}

\newcommand{\nodex}{\vertex{x}{cblue!20}}

  \scalebox{\rulescale}{
  \begin{tikzpicture}[->,node distance=12mm,n/.style={}]
    \graphbox{$R$}{68mm}{0mm}{40mm}{8mm}{-8mm}{-4mm}{
      \node [npattern] (xa)
      {\nodexa};
      \node [npattern] (xb) [right of=xa, short,xshift=-2.5mm]
      {\nodexb};
    }
    \graphbox{$G_L$}{0mm}{-9mm}{26mm}{9mm}{-1mm}{-4.8mm}{
      \node [npattern] (a)
      {\nodea};
      \node [npattern] (b) [right of=a, short, xshift=-2.5mm] {\nodeb};
    }
      \graphbox{$L$}{0mm}{0mm}{26mm}{8mm}{-1mm}{-4mm}{
      \node [npattern] (x)
      {\nodex};
       \draw [epattern, dotted] (x) to node {} (a);
    }
    \graphbox{$G_K$}{27mm}{-9mm}{40mm}{9mm}{-8mm}{-4.8mm}{
      \node [npattern] (aa)
      {\nodeaa};
      \node [npattern] (ab) [right of=aa, short, xshift=-2.5mm] {\nodeab};
      \node [npattern] (ba) [right of=ab, short, xshift=-2.5mm]
      {\nodeba};
      \node [npattern] (bb) [right of=ba, short, xshift=-2.5mm] {\nodebb};
    }
    \graphbox{$K$}{27mm}{0mm}{40mm}{8mm}{-8mm}{-4mm}{
      \node [npattern] (xa)
      {\nodexa};
      \node [npattern] (xb) [right of=xa, short, xshift=-2.5mm]
      {\nodexb};
      \draw [epattern, dotted] (xa) to node {} (aa);
      \draw [epattern, dotted] (xb) to node {} (ab);
    }
    \graphbox{$G_R$}{68mm}{-9mm}{40mm}{9mm}{-8mm}{-4.8mm}{
      \node [npattern] (aa)
      {\nodeaa};
      \node [npattern] (ab) [right of=aa, short, xshift=-2.5mm] {\nodeab};
      \node [npattern] (ba) [right of=ab, short, xshift=-2.5mm]
      {\nodeba};
      \node [npattern] (bb) [right of=ba, short, xshift=-2.5mm] {\nodebb};
    }
    \graphbox{$L'$}{0mm}{-19mm}{26mm}{8mm}{-1mm}{-4mm}{
      \node [npattern] (x)
      {\nodex};
    }
    \graphbox{$K'$}{27mm}{-19mm}{40mm}{8mm}{-8mm}{-4mm}{
      \node [npattern] (xa)
      {\nodexa};
      \node [npattern] (xb) [right of=xa, short, xshift=-2.5mm]
      {\nodexb};
    }
    \transparentgraphbox{$R'$}{68mm}{-19mm}{40mm}{8mm}{-8mm}{-4mm}{
      \node [npattern] (xa)
      {\nodexa};
      \node [npattern] (xb) [right of=xa, short, xshift=-2.5mm]
      {\nodexb};
    }
  \end{tikzpicture}
  }

}%
    \caption{Failure of Lemma~\ref{lem:uniqueness:u} for PBPO.}
    \label{fig:fail:pbpo}
\end{figure}
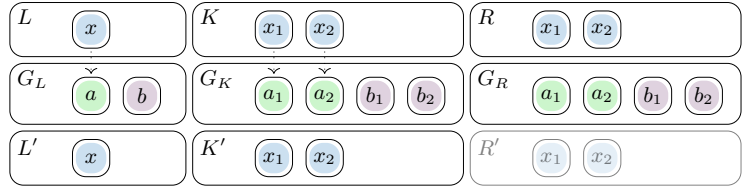

\begin{remark}
\label{remark:pbpo:u:not:unique}
In a PBPO rewrite step, not every morphism $u : K \to G_K$ satisfying $u' \circ u = t_K$ corresponds to the arrow uniquely determined by the top-left pullback. Thus Lemma~\ref{lem:uniqueness:u} does not hold for PBPO.
This can be seen in the example of a (canonical) PBPO rewrite rule and step depicted in Figure~\ref{fig:fail:pbpo}.
Because our previous notational convention breaks for this example, we indicate two morphisms by dotted arrows. The others can be inferred.

Morphism $u : K \to G_K$ (as determined by the top-left pullback) is indicated, but it can be seen that three other morphisms $v : K \to G_K$ satisfy $u' \circ v = t_K$, because every $x \in V_{K'}$ has two elements in its preimage in $G_{K}$.
\end{remark}

\begin{remark}[\pbpostrong and AGREE]
AGREE~\cite{corradini2015agree} by Corradini et al.\ is a rewriting approach closely related to PBPO. AGREE's match square can be regarded as a specialization of \pbpostrong's match square, since AGREE fixes the type morphism $t_L : L \mono L'$ of a rule as the partial map classifier for $L$. Thus, \pbpostrong can also be regarded as combining PBPO's rewriting mechanism with a generalization of AGREE's strong matching mechanism.
\end{remark}

\subsection{The Case for Strong Matching}
\label{sec:relating:case:for:strong:matching}

The two following examples serve to illustrate why we find it necessary to strengthen the matching criterion when matching is nondeterministic.

\begin{example}
    \label{ex:pbpo:remove:loop}
    In \pbpostrong, an application of the rule
    {%
\newcommand{\nodexa}{\vertex{x_1}{cblue!20}}%
\newcommand{\nodexb}{\vertex{x_2}{cblue!20}}%
\newcommand{\nodey}{\vertex{y}{cgreen!20}}%
\newcommand{\nodez}{\vertex{z}{cred!10}}%
\newcommand{\nodeu}{\vertex{u}{cpurple!25}}%
\newcommand{\nodex}{\vertex{x}{cblue!20}}%
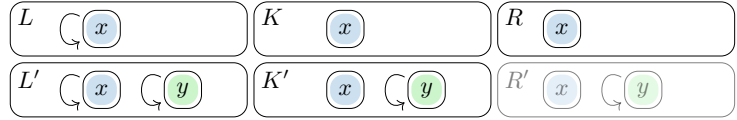
\begin{center}\vspace{0ex}
  \scalebox{\rulescale}{
  \begin{tikzpicture}[->,node distance=12mm,n/.style={}]
    \graphbox{$L$}{0mm}{0mm}{35mm}{8mm}{-4mm}{-4mm}{
      \node [npattern] (x)
      {\nodex};
      \draw [epattern,loop=180,looseness=3] (x) to node {} (x);
    }
    \graphbox{$K$}{36mm}{0mm}{35mm}{8mm}{-4mm}{-4mm}{
      \node [npattern] (x)
      {\nodex};
    }
    \graphbox{$R$}{72mm}{0mm}{35mm}{8mm}{-8mm}{-4mm}{
      \node [npattern] (x)
      {\nodex};
    }
    \graphbox{$L'$}{0mm}{-9mm}{35mm}{8mm}{-4mm}{-4mm}{
      \node [npattern] (x)
      {\nodex};
      \draw [epattern,loop=180,looseness=3] (x) to node {} (x);
      \node [npattern] (y) [right of=x] {\nodey};
      \draw [epattern,loop=180,looseness=3] (y) to node {} (y);
    }
    \graphbox{$K'$}{36mm}{-9mm}{35mm}{8mm}{-4mm}{-4mm}{
      \node [npattern] (x)
      {\nodex};
      \node [npattern] (y) [right of=x] {\nodey};
      \draw [epattern,loop=180,looseness=3] (y) to node {} (y);
    }
    \transparentgraphbox{$R'$}{72mm}{-9mm}{35mm}{8mm}{-8mm}{-4mm}{
      \node [npattern] (x)
      {\nodex};
      \node [npattern] (y) [right of=x] {\nodey};
      \draw [epattern,loop=180,looseness=3] (y) to node {} (y);
    }
  \end{tikzpicture}
  }\vspace{0ex}
\end{center}%
}%
%
    \noindent in an unlabeled graph $G_L$ removes a loop from an isolated vertex that has a single loop, and preserves everything else. In PBPO, a match is allowed to map all of $G_L$ into the component determined by vertex $\{ x \}$, so that the rule deletes all of $G_L$'s edges at once.
    (Before studying the next example, the reader is invited to consider what the effect of the PBPO rule is if $R$ and $R'$ are replaced by $L$ and $L'$, respectively.)
\end{example}

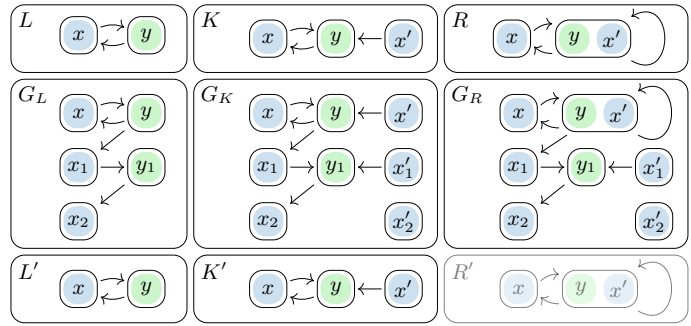
\begin{figure}[b!]
    \centering
    {%

\newcommand{\nodex}{\vertex{x}{cblue!20}}
\newcommand{\nodexa}{\vertex{x_1}{cblue!20}}
\newcommand{\nodexb}{\vertex{x_2}{cblue!20}}
\newcommand{\nodexp}{\vertex{x'}{cblue!20}}
\newcommand{\nodexap}{\vertex{x_1'}{cblue!20}}
\newcommand{\nodexbp}{\vertex{x_2'}{cblue!20}}
\newcommand{\nodey}{\vertex{y}{cgreen!20}}
\newcommand{\nodeya}{\vertex{y_1}{cgreen!20}}
\newcommand{\nodez}{\vertex{z}{cred!10}}
\newcommand{\nodeza}{\vertex{z_1}{cred!10}}
\newcommand{\nodezb}{\vertex{z_2}{cred!10}}
\newcommand{\nodezc}{\vertex{z_3}{cred!10}}
\newcommand{\nodeu}{\vertex{u}{cpurple!25}}

  \scalebox{\rulescale}{
  \begin{tikzpicture}[->,node distance=12mm,n/.style={}]
      \graphbox{$L$}{0mm}{0mm}{26mm}{10mm}{-3mm}{-4.5mm}{
      \node [npattern] (x)
      {\nodex};
      \node [npattern] (y) [right of=x, short] {\nodey};
      \draw [epattern] (x) to[bend left=20] node {} (y);
      \draw [epattern] (y) to[bend left=20] node {} (x);
    }
    \graphbox{$G_L$}{0mm}{-11mm}{26mm}{25mm}{-3mm}{-5mm}{
      \node [npattern] (x) {\nodex};
      \node [npattern] (y) [right of=x, short] {\nodey};
      \draw [epattern] (x) to[bend left=20] node {} (y);
      \draw [epattern] (y) to[bend left=20] node {} (x);

      \node [npattern] (xa) [below of=x, short, yshift=2mm] {\nodexa};
      \node [npattern] (ya) [right of=xa, short] {\nodeya};
      \node [npattern] (xb) [below of=xa, short, yshift=2mm] {\nodexb};
      
      \draw [epattern] (y) to (xa);
      \draw [epattern] (xa) to (ya);
      \draw [epattern] (ya) to (xb);
    }
    \graphbox{$G_K$}{27mm}{-11mm}{36mm}{25mm}{-7mm}{-5mm}{
      \node [npattern] (x)
      {\nodex};
      \node [npattern] (y) [right of=x, short] {\nodey};
      \draw [epattern] (x) to[bend left=20] node {} (y);
      \draw [epattern] (y) to[bend left=20] node {} (x);

      \node [npattern] (xa) [below of=x, short, yshift=2mm] {\nodexa};
      \node [npattern] (ya) [right of=xa, short] {\nodeya};
      \node [npattern] (xb) [below of=xa, short, yshift=2mm] {\nodexb};
      
      \draw [epattern] (y) to (xa);
      \draw [epattern] (xa) to (ya);
      \draw [epattern] (ya) to (xb);

      \node [npattern] (xp) [right of=y, short] {\nodexp};
      \node [npattern] (xpa) [right of=ya, short] {\nodexap};
      \node [npattern] (xpb) [below of=xpa, short, yshift=2mm] {\nodexbp};
    
      \draw [epattern] (xp) to (y);
      \draw [epattern] (xpa) to (ya);
    }
    \graphbox{$K$}{27mm}{0mm}{36mm}{10mm}{-7mm}{-5mm}{
      \node [npattern] (x)
      {\nodex};
      \node [npattern] (y) [right of=x, short] {\nodey};
      \draw [epattern] (x) to[bend left=20] node {} (y);
      \draw [epattern] (y) to[bend left=20] node {} (x);
      
      \node [npattern] (xp) [right of=y, short] {\nodexp};
      \draw [epattern] (xp) to (y);
    }
    
    \graphbox{$R$}{64mm}{0mm}{36mm}{10mm}{-8mm}{-5mm}{
      \node [npattern] (x) {\nodex};
      \node [npattern] (y) [right of=x] {\nodey \ \nodexp};
      \draw [epattern] (x) to[bend left=20] node {} (y);
      \draw [epattern] (y) to[bend left=20] node {} (x);
      
      \draw [epattern,loop=0,looseness=3] (y) to node {} (y);
    }
    
    \graphbox{$G_R$}{64mm}{-11mm}{36mm}{25mm}{-7mm}{-5mm}{
            \node [npattern] (x)
      {\nodex};
      \node [npattern] (y) [right of=x] {\nodey \ \nodexp};
      \draw [epattern] (x) to[bend left=20] node {} (y);
      \draw [epattern] (y) to[bend left=20] node {} (x);

      \node [npattern] (xa) [below of=x, short, yshift=2mm] {\nodexa};
      \node [npattern] (ya) [right of=xa, short] {\nodeya};
      \node [npattern] (xb) [below of=xa, short, yshift=2mm] {\nodexb};
      
      \draw [epattern] (y) to (xa);
      \draw [epattern] (xa) to (ya);
      \draw [epattern] (ya) to (xb);

      \node [npattern] (xpa) [right of=ya, short] {\nodexap};
      \node [npattern] (xpb) [below of=xpa, short, yshift=2mm] {\nodexbp};
    
      \draw [epattern] (xpa) to (ya);

      \draw [epattern,loop=0,looseness=3] (y) to node {} (y);
      
    }
    \graphbox{$L'$}{0mm}{-37mm}{26mm}{10mm}{-3mm}{-5mm}{
      \node [npattern] (x)
      {\nodex};
      \node [npattern] (y) [right of=x, short] {\nodey};
      \draw [epattern] (x) to[bend left=20] node {} (y);
      \draw [epattern] (y) to[bend left=20] node {} (x);
    }
    \graphbox{$K'$}{27mm}{-37mm}{36mm}{10mm}{-7mm}{-5mm}{
      \node [npattern] (x)
      {\nodex};
      \node [npattern] (y) [right of=x, short] {\nodey};
      \draw [epattern] (x) to[bend left=20] node {} (y);
      \draw [epattern] (y) to[bend left=20] node {} (x);
      
      \node [npattern] (xp) [right of=y, short] {\nodexp};
      \draw [epattern] (xp) to (y);
    }
    \transparentgraphbox{$R'$}{64mm}{-37mm}{36mm}{10mm}{-7mm}{-5mm}{
      \node [npattern] (x)
      {\nodex};
      \node [npattern] (y) [right of=x] {\nodey \ \nodexp};
      \draw [epattern] (x) to[bend left=20] node {} (y);
      \draw [epattern] (y) to[bend left=20] node {} (x);
      
      \draw [epattern,loop=0,looseness=3] (y) to node {} (y);
    }
  \end{tikzpicture}
  }

}%
    \vspace{-2ex}
    \caption{The effects of PBPO rules can be difficult to oversee.}
    \label{fig:pbpo:difficult}
\end{figure}

\begin{example}
\label{example:pbpo:spiral}
Consider the following PBPO rule, and its application to host graph $G_L$ (the morphisms are defined in the obvious way) shown in Figure~\ref{fig:pbpo:difficult}.
Intuitively, host graph $G_L$ is spiralled over the pattern of $L'$. The pullback then duplicates all elements mapped onto $x \in V_{L'}$ and any \pagebreak
incident edges directed at a node mapped into $y \in V_{L'}$. The pushout, by contrast, affects only the image of $u : K \to G_K$.
\end{example}

The two examples show how locality of transformations cannot be enforced using PBPO. They also illustrate how it can be difficult to characterize the class of host graphs $G_L$ and adherences $\alpha$ that establish a match, even for trivial left-hand sides. Finally, Example \ref{example:pbpo:spiral} in particular highlights an asymmetry that we find unintuitive: if one duplicates and then merges/extends pattern elements of $L'$, the duplication affects all elements in the $\alpha$-preimage of $t_L(L)$ (which could even consist of multiple components isomorphic to $t_L(L)$), whereas the pushout affects only $u(K) \subseteq G_K$. In \pbpostrong, by contrast, transformations of the pattern affect the pattern only, and the overall applicability of a rule is easy to understand if the context graph is relatively simple (e.g., as in Example~\ref{ex:simple:rewrite:step}).

\begin{remark}[$\Gamma$-preservation]
A locality notion has been defined for PBPO called \emph{$\Gamma$-preservation}~\cite{corradini2019pbpo}. $\Gamma$ is some subobject of $L'$, and a rewrite step  $G_L \Rightarrow_\rho^{m, \alpha} G_R$ is said to be \emph{$\Gamma$-preserving} if the $\alpha : G_L \to L'$ preimage of $\Gamma \subseteq L'$ is preserved from $G_L$ to $G_R$ (roughly meaning that this preimage is neither modified nor duplicated). Similarly, a rule is $\Gamma$-preserving if the rewrite steps it gives rise to are $\Gamma$-preserving. If one chooses $\Gamma$ to be the context graph (the right component) of $L'$ in Example~\ref{ex:pbpo:remove:loop}, then the rule, interpreted as a PBPO rule, is $\Gamma$-preserving. However, the effect of the rule is not local in our understanding of the word, since PPBO does not prevent mapping parts of the context graph of $G_L$ onto the image of $t_L$ which usually \emph{is} modified.
\end{remark}

\subsection{Modeling PBPO with \pbpostrong}
\label{sec:relating:simulating}

We will now prove that in many categories of interest (including locally small toposes), any PBPO rule can be modeled by a set of \pbpostrong rules; and even by a single rule if PBPO matches are restricted to monic matches. We do this by proving a number of novel results about PBPO.

Since any PBPO rule is equivalent to a canonical PBPO rule~\cite{corradini2019pbpo}, we restrict attention to canonical rules in this section.

\newcommand{\pbpomonicstep}[1]{\Rightarrow_{#1}^{\mono}}
\newcommand{\pbpostrongstep}[1]{\Rightarrow_{#1}^{\mathrm{SM}}}

\begin{definition}
We define the restrictions
\begin{itemize}
\item ${\pbpomonicstep{\rho}} = \{ (G,G') \in {\pbpostep{\rho}} \mid \exists m \ \alpha. \ G \Rightarrow_\rho^{m, \alpha} G' \land \text{$m$ is monic} \}$, and
\item ${\pbpostrongstep{\rho}} = 
\{ (G,G') \in {\pbpostep{\rho}} \mid 
\exists m \ \alpha. \ G \Rightarrow_\rho^{m, \alpha} G' \land \strongmatch{t_L}{m}{\alpha}
\}
$
\end{itemize}
for canonical PBPO rules $\rho$.
\end{definition}

\begin{definition}[Monic PBPO Rule]
A canonical PBPO rule $\rho$ is called \emph{monic} if its typing $t_L$ is monic.
\end{definition}

Note that monic (canonical) PBPO rules $\rho$ define a \pbpostrong rule by simply forgetting the pushout information in the rule.

\begin{proposition}
\label{prop:basic:facts:pbpo:restrictions}
If $\rho$ is monic, then ${\pbpostep{\rho}} = {\pbpomonicstep{\rho}}$ and ${\pbpostrongstep{\rho}} = {\pbpostep{\rho}^{\mathrm{PBPO}^{+}}}$. \qed
\end{proposition}

In the remainder of this section we establish two claims:
\begin{enumerate}
    \item \emph{Monic matching suffices}: 
    for any canonical PBPO rule $\rho$ and assuming certain conditions, there exists a set of PBPO rules $S$ that precisely \pagebreak models $\rho$ when restricting $S$ to monic matching, i.e., ${\pbpostep{\rho}} = {\bigcup\{ {\pbpomonicstep{\sigma}} \mid \sigma \in S \}}$ (Corollary~\ref{corollary:pbpo:monic:matching});
    \item \emph{Strong matching can be modeled through rule adaptation}: for any canonical PBPO rule $\sigma$ and assuming certain conditions, there exists a monic rule $\tau$ such that ${\pbpomonicstep{\sigma}} =
     {\pbpostrongstep{\tau}}$ (Lemma~\ref{lemma:pbpo:monic:typing}).
\end{enumerate}

Because \pbpostrong rewriting boils down to using monic PBPO rules with a strong matching policy, from these facts and conditions it follows that any PBPO rule can be modeled by a set of \pbpostrong rules (Theorem~\ref{thm:pbpo:plus:models:pbpo}).

The following definition defines a rule $\sigma$ for every factorization of a type morphism $t_L$ of a rule $\tau$.

\begin{definition}[Compacted Rule]
\label{def:pbpo:compacted:rule}
For any canonical PBPO rule $\rho$ (on the left) and factorization $t_L = t_{L_c} \mathop{\circ} e$ where $e$ is epic (note that $t_{L_c}$ is uniquely determined since $e$ is right-cancellative), the \emph{compacted rule} $\rho_e$ is defined as the lower half of the commuting diagram on the right:
\begin{center}\vspace{-1ex}
  \begin{tikzpicture}[node distance=13mm,l/.style={inner sep=.5mm},baseline=-6mm]
    \node (L) {$L$};
    \node (K) [right of=L] {$K$}; \draw [->] (K) to node [above] {$l$} (L);
    \node (LP) [below of=L] {$L'$};
      \draw [->] (L) to node [left] {$t_L$} (LP);
    \node (KP) [below of=K] {$K'$};
      \draw [->] (K) to node [right] {$t_K$} (KP);
      \draw [->] (KP) to node [below] {$l'$} (LP);
       \node at ([shift={(-4mm,-4mm)}]K) {\PB};
    \node (R) [right of=K] {$R$}; 
      \draw [->] (K) to node [above] {$r$} (R);
    \node (RP) [below of=R] {$R'$};
      \draw [->] (R) to node [right] {$t_R$} (RP);
      \draw [->] (KP) to node [below] {$r'$} (RP);
      \node at ([shift={(-4mm,4mm)}]RP) {\PO};
  \end{tikzpicture}
  \qquad
  \raisebox{2.5mm}{
  \begin{tikzpicture}[node distance=16mm,l/.style={inner sep=.5mm},baseline=-6mm,d/.style={node distance=9mm}]
    \node[opacity=0.5] (L) {$L$};
    \node[opacity=0.5] (K) [right of=L] {$K$}; \draw [->, opacity=0.5] (K) to node [above] {$l$} (L);
    \node (Lc) [below of=L,d] {$L_c$};
      \draw [->>, opacity=0.5] (L) to node [left] {$e$} (Lc);
    \node (Kc) [below of=K,d] {$K_c$};
      \draw [->, opacity=0.5] (K) to node [right] {} (Kc);
      \draw [->] (Kc) to node [below] {} (Lc);
       \node[opacity=0.5] at ([shift={(-4mm,-4mm)}]K) {\PB};
    \node[opacity=0.5] (R) [right of=K] {$R$}; 
      \draw [->, opacity=0.5] (K) to node [above] {$r$} (R);
    \node (Rc) [below of=R,d] {$R_c$};
      \draw [->, opacity=0.5] (R) to node [right] {} (Rc);
      \draw [->] (Kc) to node [below] {} (Rc);
      \node[opacity=0.5] at ([shift={(-4mm,4mm)}]Rc) {\PO};
      
    \node (LP) [below of=Lc,d] {$L'$};
      \draw [->] (Lc) to node [left] {$t_{L_c}$} (LP);
    \node (KP) [below of=Kc,d] {$K'$};
      \draw [->] (Kc) to node [right] {} (KP);
      \draw [->] (KP) to node [below] {$l'$} (LP);
       \node at ([shift={(-4mm,-4mm)}]Kc) {\PB};
    \node (RP) [below of=Rc,d] {$R'$};
      \draw [->] (KP) to node [below] {$r'$} (RP);
      \draw [->] (Rc) to node [right] {} (RP);
      \node at ([shift={(-4mm,4mm)}]RP) {\PO};
    
    \draw [->, bend right=80, opacity=0.5] (L) to node [left] {$t_L$} (LP);
    
    \draw [->, bend left=25, opacity=0.5] (K) to node [right, yshift=-3mm] {$t_K$} (KP);
    
    \draw [->, bend left=80, opacity=0.5] (R) to node [right, yshift=-3mm] {$t_R$} (RP);
  \end{tikzpicture}
  }
\end{center}
\end{definition}

\begin{proposition}
The properties implicitly asserted in Definition~\ref{def:pbpo:compacted:rule} hold. \qed
\end{proposition}

\store{lemma}{}{lem:compacted:rule:lemma}{
Let $\rho$ be a canonical PBPO rule, $G_L$ an object, and $m' \circ e : L \to G_L$ a match morphism for a mono $m'$ and epi $e$.
We have:
\begin{align*}
    {G_L \Rightarrow_\rho^{(m' \mathop{\circ} e),\alpha} G_R} \quad \iff \quad 
    {G_L \Rightarrow_{\rho_e}^{m',\alpha} G_R}\;\text{.} \tag*{\qedappendix}
\end{align*}
}

Recall that a category is locally small if the collection of morphisms between any two objects $A$ and $B$  (and so also all factorizations) forms a set.

\begin{corollary}
\label{corollary:pbpo:monic:matching}
In locally small categories in which any morphism can be factorized into an epi followed by a mono, for every canonical PBPO rule $\rho$, there exists a set of PBPO rules $S$ such that ${\pbpostep{\rho}} = {\bigcup\{ {\pbpomonicstep{\sigma}} \mid \sigma \in S \}}$. \qed
\end{corollary}

\begin{definition}[Amendable Category]
\label{def:amendable:category}
A category is $\emph{amendable}$ if for any $t_L : L \to L'$, there exists a factorization $L \stackrel{t_L'}{\mono} L'' \stackrel{\beta}{\to} L'$ of $t_L$ such that for any factorization $L \stackrel{m}{\mono} G_L \stackrel{\alpha}{\to} L'$ of $t_L$, there exists an $\alpha'$ making the diagram
\begin{center}
  \begin{tikzpicture}[node distance=12mm,l/.style={inner sep=.5mm},baseline=-6mm]
    \node (L) {$L$};
    \node (GL) [right of=L] {$G_L$};
    \node (LP) [right of=GL] {$L'$};
    \node (L2) [below of=L] {$L$};
    \node (LPP) [below of=GL] {$L''$};
    
    \draw [->] (L) to[bend left=25] node [above] {$t_L$} (LP);
    \draw [>->] (L) to node [below] {$m$} (GL);
    \draw [->] (GL) to node [below] {$\alpha$} (LP);
    
    \draw [>->] (L) to node [left] {$1_L$} (L2);
    \draw [->] (GL) to node [right] {$\alpha'$} (LPP);
    
    \draw [>->] (L2) to node [above] {$t_L'$} (LPP);
    
    \draw [->] (LPP) to [bend right=20] node [below right] {$\beta$} (LP);
  \end{tikzpicture}
\end{center}
commute.

The category is \emph{strongly amendable} if there exists a factorization of $t_L$ witnessing amendability that moreover makes the left square a pullback square.
\end{definition}
\pagebreak

Strong amendability is intimately related to the concept of \emph{materialization}~\cite{corradini2019rewriting}. Namely, if the factorization $L \stackrel{t_L'}{\mono} L'' \stackrel{\beta}{\to} L'$ of $t_L$ establishes the pullback square and is final (the $\alpha'$ morphisms not only exist, but they exist uniquely), then $\beta \circ t_L'$ is the materialization of $t_L$. In general we do not need finality, and for one statement (Lemma~\ref{lemma:pbpo:monic:typing}) we require weak amendability only.

We have the following sufficient condition for strong amendability. 

\begin{proposition}
\label{prop:pmc:strong:amendability}
If all slice categories $\catname{C}/X$ of a category $\catname{C}$ have partial map classifiers, then $\catname{C}$ is strongly amendable.
\end{proposition}
\begin{proof}
Immediate from the fact that in this case all arrows have materializations~\cite[Proposition 8]{corradini2019rewriting}.
\end{proof}

\begin{corollary}
\label{cor:toposes:strongly:amendable}
Any topos is strongly amendable.
\end{corollary}
\begin{proof}
Toposes have partial map classifiers, and any slice category of a topos is a topos. \qed
\end{proof}

\begin{lemma}
\label{lemma:pbpo:monic:typing}
In an amendable category $\mathcal{C}$, for any PBPO rule $\rho$, there exists a monic PBPO rule $\sigma$ such that ${\pbpomonicstep{\rho}} = {\pbpostep{\sigma}}$. If $\mathcal{C}$ is moreover strongly amendable, then additionally ${\pbpomonicstep{\rho}} = {\pbpostrongstep{\sigma}}$.
\end{lemma}
\begin{proof}
Given rule $\rho$ on the left
\begin{center}
  \begin{tikzpicture}[node distance=13mm,l/.style={inner sep=.5mm},baseline=-6mm]
    \node (L) {$L$};
    \node (K) [right of=L] {$K$}; \draw [->] (K) to node [above] {$l$} (L);
    \node (LP) [below of=L] {$L'$};
      \draw [->] (L) to node [left] {$t_L$} (LP);
    \node (KP) [below of=K] {$K'$};
      \draw [->] (K) to node [right] {$t_K$} (KP);
      \draw [->] (KP) to node [below] {$l'$} (LP);
       \node at ([shift={(-4mm,-4mm)}]K) {PB};
    \node (R) [right of=K] {$R$}; 
      \draw [->] (K) to node [above] {$r$} (R);
    \node (RP) [below of=R] {$R'$};
      \draw [->] (R) to node [right] {$t_R$} (RP);
      \draw [->] (KP) to node [below] {$r'$} (RP);
      \node at ([shift={(-4mm,4mm)}]RP) {PO};
  \end{tikzpicture}
  \qquad
  \raisebox{3.5mm}{
   \begin{tikzpicture}[node distance=16mm,l/.style={inner sep=.5mm},baseline=-6mm,d/.style={node distance=11mm}]
    \node (L) {$L$};
    \node (K) [right of=L] {$K$}; \draw [->] (K) to node [above] {$l$} (L);
    \node (LPP) [below of=L,d] {$L''$};
      \draw [>->] (L) to node [left] {$t_L'$} (LPP);
    \node (KPP) [below of=K,d] {$K''$};
      \draw [>->] (K) to node [right] {$t_K'$} (KPP);
      \draw [->] (KPP) to node [above] {$l''$} (LPP);
      \node at ([shift={(-4mm,-4mm)}]K) {PB};
    \node (R) [right of=K] {$R$}; 
      \draw [->] (K) to node [above] {$r$} (R);
    \node (RPP) [below of=R,d] {$R''$};
      \draw [->] (R) to node [right] {$t_R'$} (RPP);
      \draw [->] (KPP) to node [below] {$r''$} (RPP);
      \node at ([shift={(-4mm,4mm)}]RPP) {PO};

      \node[opacity=0.5] (LP) [below of=LPP,d] {$L'$};
      \node[opacity=0.5] (KP) [below of=KPP,d] {$K'$};
      \node[opacity=0.5] (RP) [below of=RPP,d] {$R'$};
      \draw [->, opacity=0.5] (KP) to node [below] {$l'$} (LP);
      \draw [->, opacity=0.5] (KP) to node [below] {$r'$} (RP);

      \draw [->, opacity=0.5] (LPP) to node [left] {$\beta$} (LP);
      \draw [->, opacity=0.5] (KPP) to node [right] {} (KP);
      \draw [->, opacity=0.5] (RPP) to node [right] {} (RP);

      \node[opacity=0.5] at ([shift={(-4mm,-4mm)}]KPP) {PB};
      \node[opacity=0.5] at ([shift={(-4mm,4mm)}]RP) {PO};

      \draw[->, opacity=0.5] (L) to[looseness=1.4, bend right=70] node [left] {$t_L$} (LP);
  \end{tikzpicture}
  }
\end{center}
we can construct rule $\sigma$ as the upper half of the diagram on the right,
where $L \stackrel{t_L'}{\mono} L'' \stackrel{\beta}{\to} L'$ is the factorization of $t_L$ witnessing strong amendability. Then the first claim ${\pbpomonicstep{\rho}} = {\pbpostep{\sigma}}$ follows by considering the commuting diagram
\begin{center}
  \begin{tikzpicture}[node distance=16mm,l/.style={inner sep=.5mm},baseline=-6mm,d/.style={node distance=11mm}]
    \node (L) {$L$};
    \node (K) [right of=L] {$K$}; 
      \draw [->] (K) to node [above] {$l$} (L);
    \node (R) [right of=K] {$R$}; 
      \draw [->] (K) to node [above] {$r$} (R);

    \node (GL) [below of=L,d] {$G_L$};
      \draw [>->] (L) to node [right] {$m$} (GL);
    \node (GK) [below of=K,d] {$G_K$};
      \draw [->] (K) to node [right] {} (GK);
      \draw [->] (GK) to node [below] {} (GL);
       \node at ([shift={(-4mm,-4mm)}]K) {PB};
    \node (GR) [below of=R,d] {$G_R$};
      \draw [->] (R) to node [right] {} (GR);
      \draw [->] (GK) to node [below] {} (GR);
      \node at ([shift={(-4mm,4mm)}]GR) {PO};
      
    \node (L'') [below of=GL,d] {$L''$};
    \node (K'') [below of=GK,d] {$K''$};
      \draw [->] (K'') to node [above] {$l''$} (L'');
      \draw [->] (GK) to (K'');
      \node at ([shift={(-4mm,-4mm)}]K'') {PB};
    \node (R'') [below of=GR,d] {$R''$};
      \draw [->] (K'') to node [below] {$r''$} (R'');
      \draw [->] (GR) to (R'');
      \node at ([shift={(-4mm,4mm)}]R'') {PO};
      
    \node (L') [below of=L'',d] {$L'$};
      \draw [->] (L'') to node [right] {$\beta$} (L');
    \node (K') [below of=K'',d] {$K'$};
      \draw [->] (K'') to node [right] {} (K');
      \draw [->] (K') to node [below] {$l'$} (L');
       \node at ([shift={(-4mm,-4mm)}]GK) {PB};
    \node (R') [below of=R'',d] {$R'$};
      \draw [->] (K') to node [below] {$r'$} (R');
      \draw [->] (R'') to node [right] {} (R');
      \node at ([shift={(-4mm,4mm)}]R') {PO};
    
    \draw [->,out=180,in=180,looseness=2] (L) to node [left] {$t_L$} (L');
    \draw [->] (GL) to node [right] {$\alpha'$} (L'');
    \draw [>->, bend right=30] (L) to node [left,pos=0.2] {$t_L'$} (L'');
    \draw [->, bend right=30] (GL) to node [left,pos=0.8] {$\alpha$} (L');

    \node (leftL1) [left of=GL] {$L$};
    \node (leftL2) [left of=L''] {$L$};
      \draw [>->] (leftL1) to node [above,pos=0.4] {$m$} (GL);
      \draw [>->] (leftL2) to node [above,pos=0.4] {$t_L'$} (L'');
      \draw [>->] (leftL1) to node [left] {$1_L$} (leftL2);
      \node at ($(leftL1)!.3!(L'')$) {$\dagger$};
  \end{tikzpicture}
\end{center}
and the second claim ${\pbpomonicstep{\rho}} = {\pbpostrongstep{\sigma}}$ for strongly amendable categories follows by observing that the square marked by $\dagger$ is a pullback square.~\qed
\end{proof}
\pagebreak

\store{theorem}{}{thm:pbpo:plus:models:pbpo}{
In locally small, strongly amendable categories in which every morphism $f$ can be factored into an epi $e$ followed by a mono $m$, any PBPO rule $\rho$ can be modeled by a set of \pbpostrong rules. \qedappendix
}

\begin{corollary}
In any locally small topos, any PBPO rule $\rho$ can be modeled by a set of \pbpostrong rules.
\end{corollary}
\begin{proof}
By Corollary~\ref{cor:toposes:strongly:amendable} and the fact that toposes are epi-mono factorizable. \qed
\end{proof}

\section{Category \GraphLattice{(\labels,\leq)}}
\label{sec:category:graph:lattice}

Unless one employs a meta-notation or restricts to unlabeled graphs, as we did in Section~\ref{sec:pbpostrong}, it is sometimes impractical to use \pbpostrong in the category \Graph. The following example illustrates the problem.
\medskip

\noindent
\begin{minipage}[c]{0.78\textwidth}
\begin{example}\label{ex:pbpostrong:cumbersome:in:graph}
    Suppose the set of labels is $\labels = \{0,1\}$. To be able to injectively match pattern
    $
    L = %
    \begin{tikzcd}[column sep=0.5cm]
    0 \arrow[r, "1"] & 0
    \end{tikzcd}
    $
    in any context, one must inject it into the type graph $L'$ shown on the right
    in which every dotted loop represents two edges (one for each label), and every dotted non-loop represents four edges (one for each label, in either direction). In general, to allow any context, one needs  to include $|\labels|$ additional vertices in $L'$, and $|\labels|$ complete graphs over $V_{L'}$.
\end{example}
\end{minipage}\hfill%
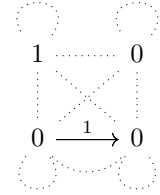
\begin{minipage}[c]{0.2\textwidth}
    \hfill
    \begin{tikzcd}
    1 \arrow[d, no head, dotted] \arrow[rd, no head, dotted] \arrow[r, no head, dotted] \arrow[no head, dotted, loop, distance=2em, in=125, out=55] & 0 \arrow[d, no head, dotted] \arrow[ld, no head, dotted] \arrow[no head, dotted, loop, distance=2em, in=125, out=55] \\
    0 \arrow[r, "1"] \arrow[r] \arrow[no head, dotted, loop, distance=2em, in=305, out=235] \arrow[r, no head, dotted, bend right=49]               & 0 \arrow[no head, dotted, loop, distance=2em, in=305, out=235]                                                      
    \end{tikzcd}\hspace*{-4mm}%
\end{minipage}

Beyond this example, and less easily alleviated with meta-notation, in \Graph it is impractical or  impossible to express rules that involve (i)~arbitrary labels (or classes of labels) in the application condition; 
(ii)~relabeling; or (iii)~allowing and 
capturing arbitrary subgraphs (or classes of subgraphs) around a match graph.
As we will discuss in Section~\ref{sec:discussion}, these features have been non-trivial to express in general for algebraic graph rewriting approaches.

We define a category which allows flexibly addressing all of these issues.

\begin{definition}[Complete Lattice]
A \emph{complete lattice} $(\labels, \leq)$ is a poset such that all subsets $S$ of $\labels$ have a supremum (join) $\bigvee S$ and an infimum (meet) $\bigwedge S$. 
\end{definition}

\begin{definition}[\GraphLattice{(\labels,\leq)}]
For a complete lattice $(\labels, \leq)$, we define the category \GraphLattice{(\labels,\leq)}, where objects are graphs labeled from $\labels$, and  arrows are graph premorphisms $\phi : G \to G'$ that satisfy $\lbl_G(x) \leq \lbl_{G'}(\phi(x))$ for all $x \in V_G \cup E_G$. %
\end{definition}

In terms of graph structure, the pullbacks and pushouts in \GraphLattice{(\labels,\leq)} are the usual pullbacks and pushouts in $\Graph$. The only difference is that the labels that are identified by respectively the cospan and span are replaced by their meet and join, respectively.

The sufficient condition of Proposition~\ref{prop:pmc:strong:amendability} does not hold in \GraphLattice{(\labels,\leq)}. Nonetheless, we have the following result.

\store{lemma}{}{lem:graph:lattice:strongly:amendable}{
  \GraphLattice{(\labels,\leq)} is strongly amendable.\qedappendix
}

One very simple but extremely useful complete lattice is the following.

\begin{definition}[Flat Lattice]
Let $\labels^{\bot,\top} = \labels \uplus \{ \bot, \top \}$. We define the \emph{flat lattice} induced by $\labels$ as the smallest poset 
\pagebreak
$(\labels^{\bot,\top}, {\leq})$, which has $\bot$ as a global minimum and $\top$ as a global maximum (so in particular, the elements of $\labels$ are incomparable). In this context, we refer to $\labels$ as the \emph{base label set}.
\end{definition}

One feature flat lattices provide is a kind of ``wildcard element'' $\top$.
\medskip

\noindent
\begin{minipage}[c]{0.65\textwidth}
    \begin{example}[Wildcards]
        Using flat lattices, $L'$ of Example~\ref{ex:pbpostrong:cumbersome:in:graph} can be fully expressed for any base label set $\mathcal{L} \ni 0,1$ as shown on the right (node identities are omitted).
        The visual syntax and naming shorthands of PGR~\cite{overbeek2020patch} (or variants thereof) could be leveraged to simplify the notation further.
    \end{example}
\end{minipage}\hfill%
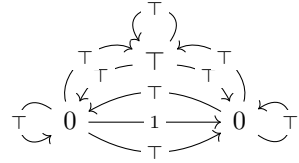
\begin{minipage}[c]{0.31\textwidth}
    \hfill
        \begin{tikzcd}[row sep=0.3cm, column sep=0.6cm]
                                                                                                                                                                                              & \top \arrow["\top" description, loop, distance=2em, in=125, out=55] \arrow[ld, "\top" description, bend right] \arrow[rd, "\top" description, bend left] &                                                                                                                                                           \\
        0 \arrow[rr, "1" description] \arrow[ru, "\top" description, bend left=67] \arrow[rr, "\top" description, bend right] \arrow["\top" description, loop, distance=2em, in=215, out=145] &                                                                                                                                                          & 0 \arrow[lu, "\top" description, bend right=67] \arrow[ll, "\top" description, bend right] \arrow["\top" description, loop, distance=2em, in=35, out=325]
        \end{tikzcd}
\end{minipage}
\medskip

As the following example illustrates, the expressive power of a flat lattice stretches beyond wildcards: it also enables relabeling of graphs. (Henceforth, we will depict a node $x$ with label $u$ as $x^u$.)

\begin{example}[Relabeling]
\label{example:hard:overwriting}
  As vertex labels we employ the flat lattice induced by the set $\{\, a,b,c,\ldots \,\}$, and assume edges are unlabeled for notational simplicity.
  The following diagram displays a rule ($L,L',K,K',R$)  for overwriting an arbitrary vertex's label with $c$, in any context. We include an application to an example graph in the middle row:
  {%

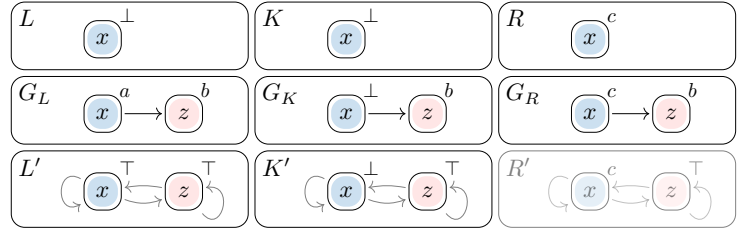
\newcommand{\nodexa}{\vertex{x_1}{cblue!20}}
\newcommand{\nodexb}{\vertex{x_2}{cblue!20}}
\newcommand{\nodexc}{\vertex{x_3}{cblue!20}}
\newcommand{\nodexd}{\vertex{x_4}{cblue!20}}

\newcommand{\nodea}{\vertex{a}{cgreen!20}}
\newcommand{\nodeb}{\vertex{b}{cpurple!25}}

\newcommand{\nodeaa}{\vertex{a_1}{cgreen!20}}
\newcommand{\nodeab}{\vertex{a_2}{cgreen!20}}
\newcommand{\nodeba}{\vertex{b_1}{cpurple!25}}
\newcommand{\nodebb}{\vertex{b_2}{cpurple!25}}

\newcommand{\nodex}{\vertex{x}{cblue!20}}

\newcommand{\nodez}{\vertex{z}{cred!10}}

\begin{center}
  \scalebox{\rulescale}{
  \begin{tikzpicture}[->,node distance=12mm,n/.style={}]
    \graphbox{$L$}{0mm}{0mm}{35mm}{10mm}{-4mm}{-5.5mm}{
      \node [npattern] (x)
      {\nodex};
       \annotate{x}{$\bot$};
    }
    \graphbox{$K$}{36mm}{0mm}{35mm}{10mm}{-4mm}{-5.5mm}{
      \node [npattern] (x)
      {\nodex};
      \annotate{x}{$\bot$};
    }
    \graphbox{$R$}{72mm}{0mm}{35mm}{10mm}{-4mm}{-5.5mm}{
      \node [npattern] (x)
      {\nodex};
       \annotate{x}{$c$};
    }
    \graphbox{$G_L$}{0mm}{-11mm}{35mm}{10mm}{-4mm}{-5.5mm}{
      \node [npattern] (x)
      {\nodex};
       \annotate{x}{$a$};
       \node [npattern] (z) [right of=x] {\nodez};
       \annotate{z}{$b$};
        \draw [epattern] (x) to node {} (z);
    }
    \graphbox{$G_K$}{36mm}{-11mm}{35mm}{10mm}{-4mm}{-5.5mm}{
      \node [npattern] (x)
      {\nodex};
      \annotate{x}{$\bot$};
      \node [npattern] (z) [right of=x] {\nodez};
       \annotate{z}{$b$};
        \draw [epattern] (x) to node {} (z);
    }
    \graphbox{$G_R$}{72mm}{-11mm}{35mm}{10mm}{-4mm}{-5.5mm}{
      \node [npattern] (x)
      {\nodex};
      \annotate{x}{$c$};
      \node [npattern] (z) [right of=x] {\nodez};
       \annotate{z}{$b$};
        \draw [epattern] (x) to node {} (z);
    }
    \graphbox{$L'$}{0mm}{-22mm}{35mm}{11.3mm}{-4mm}{-6mm}{
      \node [npattern] (x)
      {\nodex};
       \annotate{x}{$\top$};
        \draw [eset,loop=180,looseness=3] (x) to node {} (x);
        \node [npattern] (z) [right of=x] {\nodez};
       \annotate{z}{$\top$};
        \draw [eset] (x) to [bend right=20] node {} (z);
         \draw [eset] (z) to [bend right=10] node {} (x);
         \draw [eset,loop=-20,looseness=3] (z) to node {} (z);
    }
    \graphbox{$K'$}{36mm}{-22mm}{35mm}{11.3mm}{-4mm}{-6mm}{
      \node [npattern] (x)
      {\nodex};
       \annotate{x}{$\bot$};
        \draw [eset,loop=180,looseness=3] (x) to node {} (x);
        \node [npattern] (z) [right of=x] {\nodez};
       \annotate{z}{$\top$};
        \draw [eset] (x) to [bend right=20] node {} (z);
         \draw [eset] (z) to [bend right=10] node {} (x);
         \draw [eset,loop=-20,looseness=3] (z) to node {} (z);
    }
    \transparentgraphbox{$R'$}{72mm}{-22mm}{35mm}{11.3mm}{-4mm}{-6mm}{
      \node [npattern] (x)
      {\nodex};
       \annotate{x}{$c$};
        \draw [eset,loop=180,looseness=3] (x) to node {} (x);
        \node [npattern] (z) [right of=x] {\nodez};
       \annotate{z}{$\top$};
        \draw [eset] (x) to [bend right=20] node {} (z);
         \draw [eset] (z) to [bend right=10] node {} (x);
         \draw [eset,loop=-20,looseness=3] (z) to node {} (z);
    }
  \end{tikzpicture}
  }
\end{center}

}%
  \noindent
  The example demonstrates how (i)~labels in $L$ serve as lower bounds for matching, (ii)~labels in $L'$ serve as upper bounds for matching, (iii)~labels in $K'$ can be used to decrease matched labels (so in particular, $\bot$ ``instructs'' to ``erase'' the label and overwrite it with $\bot$, and $\top$ ``instructs'' to preserve labels), and (iv)~labels in $R$ can be used to increase labels.
\end{example}

Complete lattices also support modeling sorts.
\medskip

\noindent%
\begin{minipage}[c]{.71\textwidth}%
  \useparinfo{parinfo}%
  \begin{example}[Sorts]
  Let $p_1, p_2,\ldots \in \mathbb{P}$ be a set of processes and $d_1,d_2,\ldots \in \mathbb{D}$ a set of data elements. Assume a complete lattice over labels $\mathbb{P} \cup \mathbb{D} \cup \{ \mathbb{P}, \mathbb{D}, \rhd, @ \}$, arranged as in the diagram on the right.
  
  Moreover, assume that the vertices $x,y,\ldots$ in the graphs of interest are labeled with a $p_i$ or $d_i$, and that edges are labeled with a $\rhd$ or $@$. In such a graph,
  \end{example}
\end{minipage}\hfill%
\begin{minipage}[c]{.22\textwidth}
    $\forall i \in \mathbb{N}:$\\[.5ex]
    \begin{tikzcd}[column sep=1.5mm,row sep=0.2cm]
                          & \top                                                                       &                 &               \\
    \mathbb{P} \arrow[ru] & \mathbb{D} \arrow[u]                                                       & \rhd \arrow[lu] & @ \arrow[llu] \\
    p_i \arrow[u]         & d_i \arrow[u]                                                              &                 &               \\
                          & \bot \arrow[ruu, bend right] \arrow[u] \arrow[lu] \arrow[rruu, bend right] &                 &              
    \end{tikzcd}
    \vspace{-1.5ex}
\end{minipage}
    \begin{itemize}
        \item 
          an edge $x^{d_i} \xrightarrow{@} y^{p_j}$ encodes that process $p_j$ holds a local copy of datum $d_i$ ($x$ will have no other connections); and
        \item 
          a chain of edges $x^{p_i} \xrightarrow{\rhd} y^{d_k} \xrightarrow{\rhd} z^{d_l} \xrightarrow{\rhd} \cdots \xrightarrow{\rhd} u^{p_j}$ encodes a directed FIFO channel from process $p_i$ to process $p_j \neq p_i$, containing a sequence of elements $d_k, d_l, \ldots$. An empty channel is modeled as $x^{p_i} \xrightarrow{\rhd} u^{p_j}$.
          \pagebreak 
    \end{itemize}
    Receiving a datum through an incoming channel (and storing it locally) can be modeled using the following rule:
    {%

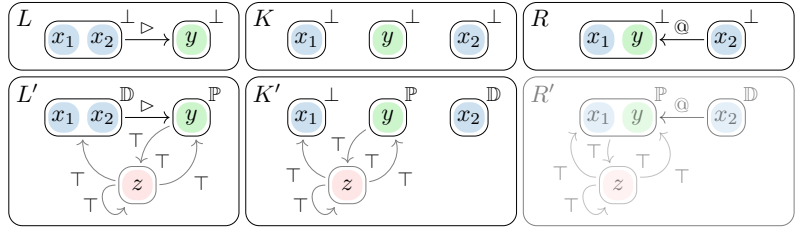
\newcommand{\nodexa}{\vertex{x_1}{cblue!20}}
\newcommand{\nodexb}{\vertex{x_2}{cblue!20}}
\newcommand{\nodexc}{\vertex{x_3}{cblue!20}}
\newcommand{\nodexd}{\vertex{x_4}{cblue!20}}

\newcommand{\nodea}{\vertex{a}{cgreen!20}}
\newcommand{\nodeb}{\vertex{b}{cpurple!25}}

\newcommand{\nodeaa}{\vertex{a_1}{cgreen!20}}
\newcommand{\nodeab}{\vertex{a_2}{cgreen!20}}
\newcommand{\nodeba}{\vertex{b_1}{cpurple!25}}
\newcommand{\nodebb}{\vertex{b_2}{cpurple!25}}

\newcommand{\nodex}{\vertex{x}{cblue!20}}

\newcommand{\nodez}{\vertex{z}{cred!10}}

\newcommand{\nodey}{\vertex{y}{cgreen!20}}

\begin{center}
  \scalebox{\rulescale}{
  \begin{tikzpicture}[->,node distance=12mm,n/.style={}]
    \graphbox{$L$}{0mm}{0mm}{34mm}{10mm}{-6mm}{-5.5mm}{
      \node [npattern] (x)
      {\nodexa \ \nodexb};
       \annotate{x}{$\bot$};
       \node [npattern] (y) [right of=x,xshift=4mm] {\nodey};
       \annotate{y}{$\bot$};
       \draw [epattern] (x) to node [above,label,inner sep=0.5mm] {$\rhd$} (y);
    }
    \graphbox{$L'$}{0mm}{-11mm}{34mm}{22mm}{-6mm}{-6mm}{
      \node [npattern] (x)
      {\nodexa \ \nodexb};
       \annotate{x}{$\mathbb{D}$};
        \node [npattern] (y) [right of=x, xshift=4mm] {\nodey};
       \annotate{y}{$\mathbb{P}$};
       \draw [epattern] (x) to node [above,label,inner sep=0.5mm] {$\rhd$} (y);
       \node [nset] (z) at ($(x)!0.5!(y)$) [yshift=-10mm] {\nodez};
       \annotate{z}{$\top$};
       \draw [eset,loop=-150,looseness=3] (z) to node [left,label,inner sep=0.5mm] {$\top$} (z);
       \draw [eset] (z) to [bend left=40] node [below left,label,inner sep=0.5mm] {$\top$} (x);
       \draw [eset] (z) to[bend right=50] node [below right,label,inner sep=0.5mm] {$\top$} (y);
       \draw [eset] (y) to[bend right=30] node [ left,label,inner sep=0.5mm] {$\top$} (z);
    }
    \graphbox{$K'$}{35mm}{-11mm}{40mm}{22mm}{-11mm}{-6mm}{
      \node [npattern] (x)
      {\nodexa};
       \annotate{x}{$\bot$};
        \node [npattern] (y) [right of=x] {\nodey};
       \annotate{y}{$\mathbb{P}$};
       \node [nset] (z) at ($(x)!0.5!(y)$) [yshift=-10mm] {\nodez};
       \annotate{z}{$\top$};
       \draw [eset,loop=-150,looseness=3] (z) to node [left,label,inner sep=0.5mm] {$\top$} (z);
       \draw [eset] (z) to [bend left=40] node [below left,label,inner sep=0.5mm] {$\top$} (x);
       \draw [eset] (z) to[bend right=50] node [below right,label,inner sep=0.5mm] {$\top$} (y);
       \draw [eset] (y) to[bend right=30] node [ left,label,inner sep=0.5mm] {$\top$} (z);
       \node [npattern] (x2) [right of=y] {$\nodexb$};
       \annotate{x2}{$\mathbb{D}$};
    }
    \graphbox{$K$}{35mm}{0mm}{40mm}{10mm}{-11mm}{-5.5mm}{
      \node [npattern] (x)
      {\nodexa};
       \annotate{x}{$\bot$};
        \node [npattern] (y) [right of=x] {\nodey};
       \annotate{y}{$\bot$};
       \node [npattern] (x2) [right of=y] {$\nodexb$};
       \annotate{x2}{$\bot$};
    }
    \graphbox{$R$}{76mm}{0mm}{40mm}{10mm}{-6mm}{-5.5mm}{
      \node [npattern] (x)
      {\nodexa \ \nodey};
       \annotate{x}{$\bot$};
       \node [npattern] (x2) [right of=x, xshift=4mm] {$\nodexb$};
       \annotate{x2}{$\bot$};
       \draw [epattern] (x2) to node [above,label,inner sep=0.5mm] {$@$} (x);
    }
    \transparentgraphbox{$R'$}{76mm}{-11mm}{40mm}{22mm}{-6mm}{-6mm}{
      \node [npattern] (x)
      {\nodexa \ \nodey};
       \annotate{x}{$\mathbb{P}$};
       \node [nset] (z) [yshift=-10mm] {\nodez};
       \annotate{z}{$\top$};
       \draw [eset,loop=-150,looseness=3] (z) to node [left,label,inner sep=0.5mm] {$\top$} (z);
       \draw [eset] (z) to [bend left=70] node [below left,label,inner sep=0.5mm] {$\top$} (x);
       \draw [eset] (z) to[bend right=70,looseness=1.5] node [below right,label,inner sep=0.5mm] {$\top$} (x);
       \draw [eset] (x) to[bend right=20] node [ left,label,inner sep=0.5mm] {$\top$} (z);
        \node [npattern] (x2) [right of=x, xshift=4mm] {$\nodexb$};
       \annotate{x2}{$\mathbb{D}$};
       \draw [epattern] (x2) to node [above,label,inner sep=0.5mm] {$@$} (x);
    }
  \end{tikzpicture}
  }
\end{center}

}%
    The rule illustrates how sorts can improve readability and provide type safety. For instance, the label $\mathbb{D}$ in $L'$ prevents empty channels from being matched. More precisely, always the last element $d$ of a non-empty channel is matched. $K'$ duplicates the node holding $d$: for duplicate $x_1$, the label is forgotten but the connection to the context retained, allowing it to be fused with $y$; and for $x_2$, the connection is forgotten but the label retained, allowing it to be connected to $y$ as an otherwise isolated node. 

Finally, a very powerful feature provided by the coupling of \pbpostrong and \GraphLattice{(\labels,\leq)} is the ability to model a general notion of variable. This is achieved by using multiple context nodes in $L'$ (i.e., nodes not in the image of $t_L$).

\newcommand{\upcxtnode}{\mathcal{C}}

\begin{example}[Variables]
The rule
$f(g(x),y) \to h(g(x), g(y), x)$
on ordered trees
can be precisely modeled in \pbpostrong by the rule
{%

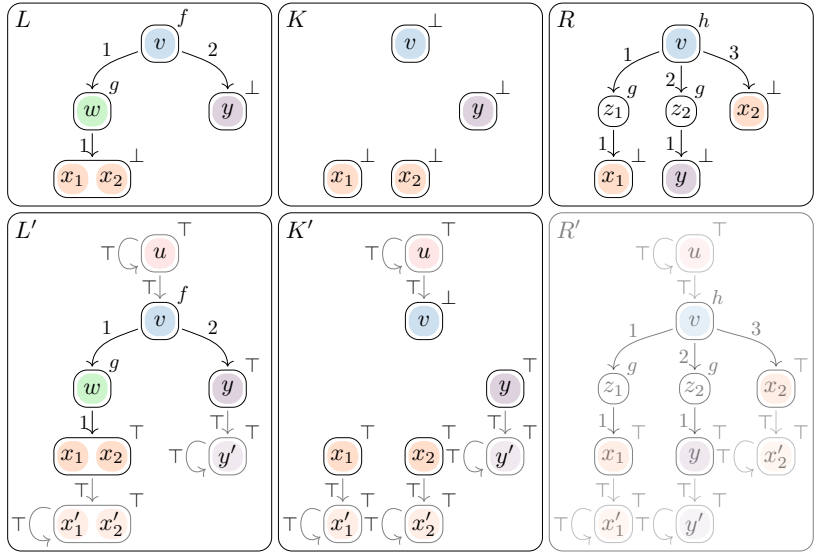
\newcommand{\noder}{\vertex{r}{cred!20}}
\newcommand{\nodex}{\vertex{v}{cblue!20}}
\newcommand{\nodey}{\vertex{w}{cgreen!20}}
\newcommand{\nodeu}{\vertex{u}{cred!10}}
\newcommand{\nodeva}{\vertex{x_1}{corange!20}}
\newcommand{\nodevb}{\vertex{x_2}{corange!20}}
\newcommand{\nodevap}{\vertex{x'_1}{corange!10}}
\newcommand{\nodevbp}{\vertex{x'_2}{corange!10}}
\newcommand{\nodew}{\vertex{y}{cpurple!25}}
\newcommand{\nodewp}{\vertex{y'}{cpurple!12}}

\begin{center}
  \scalebox{\rulescale}{
  \begin{tikzpicture}[->,node distance=12mm,n/.style={}]
    \graphbox{$L$}{0mm}{0mm}{39mm}{30mm}{3mm}{-6mm}{
      \node [npattern] (x) [short] {\nodex}; \annotate{x}{$f$};
      \node (c) [below of=x,short] {};
      \node [npattern] (y) [left of=c,short] {\nodey}; \annotate{y}{$g$};
      \node [npattern] (w) [right of=c,short] {\nodew}; \annotate{w}{$\bot$};
      \node [npattern] (v) [below of=y,short]
      {\nodeva \ \nodevb}; \annotate{v}{$\bot$};
      
      \draw [epattern] (x) to[out=-160,in=90] node [above,label,inner sep=1mm] {$1$} (y);
      \draw [epattern] (y) to node [left,label] {$1$} (v);
      \draw [epattern] (x) to[out=-20,in=90] node [above,label,inner sep=1mm] {$2$} (w);
    }
    \graphbox{$K$}{40mm}{0mm}{39mm}{30mm}{0mm}{-6mm}{
      \node [npattern] (x) [short] {\nodex}; \annotate{x}{$\bot$};
      \node (c) [below of=x,short] {};
      \node (y) [left of=c,short] {}; 
      \node [npattern] (w) [right of=c,short] {\nodew}; \annotate{w}{$\bot$};

      \node [npattern] (v1) [below of=y,short] {\nodeva}; \annotate{v1}{$\bot$};
      \node [npattern] (v2) [right of=v1,short] {\nodevb}; \annotate{v2}{$\bot$};
    }
    \graphbox{$R$}{80mm}{0mm}{39mm}{30mm}{0mm}{-6mm}{
      \node [npattern] (x) [short] {\nodex}; \annotate{x}{$h$};
      \node [npattern] (z2) [below of=x,short] {$z_2$}; \annotate{z2}{$g$};
      \node [npattern] (z1) [left of=z2,short] {$z_1$}; \annotate{z1}{$g$};
      \node [npattern] (v2) [right of=z2,short] {\nodevb}; \annotate{v2}{$\bot$};
      \node [npattern] (w) [below of=z2,short] {\nodew}; \annotate{w}{$\bot$};
      \node [npattern] (v1) [below of=z1,short] {\nodeva}; \annotate{v1}{$\bot$};

      \draw [epattern] (x) to[out=-160,in=90] node [above,label,inner sep=.5mm] {$1$} (z1);
      \draw [epattern] (x) to node [left,label,inner sep=0.5mm] {$2$} (z2);
      \draw [epattern] (x) to[out=-20,in=90] node [above,label,inner sep=0.5mm] {$3$} (v2);
      \draw [epattern] (z1) to node [left,label,inner sep=0.5mm] {$1$} (v1);
      \draw [epattern] (z2) to node [left,label,inner sep=0.5mm] {$1$} (w);
    }
    \graphbox{$L'$}{0mm}{-31mm}{39mm}{50mm}{3mm}{-6mm}{
      \node [nset] (u) [short] {\nodeu}; \annotate{u}{$\top$};
      \node [npattern] (x) [below of=u,short] {\nodex}; \annotate{x}{$f$};
      \node (c) [below of=x,short] {};
      \node [npattern] (y) [left of=c,short] {\nodey}; \annotate{y}{$g$};
      \node [npattern] (w) [right of=c,short] {\nodew}; \annotate{w}{$\top$};
      \node [npattern] (v) [below of=y,short] {\nodeva \ \nodevb}; \annotate{v}{$\top$};
      \node [nset] (v') [below of=v,short] {\nodevap \ \nodevbp}; \annotate{v'}{$\top$};
      \node [nset] (w') [below of=w,short] {\nodewp}; \annotate{w'}{$\top$};
      
      \draw [eset] (u) to node [left,label] {$\top$} (x);
      \draw [eset,loop=180,looseness=3] (u) to node [left,label] {$\top$} (u);
      \draw [epattern] (x) to[out=-160,in=90] node [above,label,inner sep=1mm] {$1$} (y);
      \draw [epattern] (y) to node [left,label] {$1$} (v);
      \draw [eset] (v) to node [left,label] {$\top$} (v');
      \draw [eset,thinloop=180,looseness=2.5] (v') to node [left,label] {$\top$} (v');
      \draw [epattern] (x) to[out=-20,in=90] node [above,label,inner sep=1mm] {$2$} (w);
      \draw [eset] (w) to node [left,label] {$\top$} (w');
      \draw [eset,loop=180,looseness=3] (w') to node [left,label] {$\top$} (w');
    }
    \graphbox{$K'$}{40mm}{-31mm}{39mm}{50mm}{2mm}{-6mm}{
      \node [nset] (u) [short] {\nodeu}; \annotate{u}{$\top$};
      \node [npattern] (x) [below of=u,short] {\nodex}; \annotate{x}{$\bot$};
      \node (c) [below of=x,short] {};
      \node (y) [left of=c] {};
      \node [npattern] (w) [right of=c] {\nodew}; \annotate{w}{$\top$};
      \node [nset] (w') [below of=w,short] {\nodewp}; \annotate{w'}{$\top$};
      \node [npattern] (v1) [below of=y,short] {\nodeva}; \annotate{v1}{$\top$};
      \node [npattern] (v2) [right of=v1] {\nodevb}; \annotate{v2}{$\top$};
      \node [nset] (v1') [below of=v1,short] {\nodevap}; \annotate{v1'}{$\top$};
      \node [nset] (v2') [below of=v2,short] {\nodevbp}; \annotate{v2'}{$\top$};
      
      \draw [eset] (u) to node [left,label] {$\top$} (x);
      \draw [eset] (v1) to node [left,label] {$\top$} (v1');
      \draw [eset,loop=180,looseness=3] (v1') to node [left,label] {$\top$} (v1');
      \draw [eset] (v2) to node [left,label] {$\top$} (v2');
      \draw [eset,loop=180,looseness=3] (v2') to node [left,label] {$\top$} (v2');
      \draw [eset,loop=180,looseness=3] (u) to node [left,label] {$\top$} (u);
      \draw [eset] (w) to node [left,label] {$\top$} (w');
      \draw [eset,loop=180,looseness=3] (w') to node [left,label] {$\top$} (w');
    }
    \transparentgraphbox{$R'$}{80mm}{-31mm}{39mm}{50mm}{2mm}{-6mm}{
      \node [nset] (u) [short] {\nodeu}; \annotate{u}{$\top$};
      \node [npattern] (x) [below of=u,short] {\nodex}; \annotate{x}{$h$};
      \node [npattern] (z2) [below of=x,short] {$z_2$}; \annotate{z2}{$g$};
      \node [npattern] (z1) [left of=z2] {$z_1$}; \annotate{z1}{$g$};
      \node [npattern] (v2) [right of=z2] {\nodevb}; \annotate{v2}{$\top$};
      \node [npattern] (w) [below of=z2,short] {\nodew}; \annotate{w}{$\top$};
      \node [nset] (w') [below of=w,short] {\nodewp}; \annotate{w'}{$\top$};
      \node [npattern] (v1) [below of=z1,short] {\nodeva}; \annotate{v1}{$\top$};
      \node [nset] (v1') [below of=v1,short] {\nodevap}; \annotate{v1'}{$\top$};
      \node [nset] (v2') [below of=v2,short] {\nodevbp}; \annotate{v2'}{$\top$};
      
      \draw [eset] (u) to node [left,label] {$\top$} (x);
      \draw [eset] (v1) to node [left,label] {$\top$} (v1');
      \draw [eset,loop=180,looseness=3] (v1') to node [left,label] {$\top$} (v1');
      \draw [eset] (v2) to node [left,label] {$\top$} (v2');
      \draw [eset,loop=180,looseness=3] (v2') to node [left,label] {$\top$} (v2');
      \draw [eset,loop=180,looseness=3] (u) to node [left,label] {$\top$} (u);
      \draw [eset] (w) to node [left,label] {$\top$} (w');
      \draw [eset,loop=180,looseness=3] (w') to node [left,label] {$\top$} (w');

      \draw [epattern] (x) to[out=-160,in=90] node [above,label,inner sep=.5mm] {$1$} (z1);
      \draw [epattern] (x) to node [left,label,inner sep=0.5mm] {$2$} (z2);
      \draw [epattern] (x) to[out=-20,in=90] node [above,label,inner sep=0.5mm] {$3$} (v2);
      \draw [epattern] (z1) to node [left,label,inner sep=0.5mm] {$1$} (v1);
      \draw [epattern] (z2) to node [left,label,inner sep=0.5mm] {$1$} (w);
    }
  \end{tikzpicture}
  }
\end{center}

}%
\noindent if one restricts the set of rewritten graphs to straightforward representations of trees: nodes are labeled by symbols, and edges are labeled by $n \in \mathbb{N}$, the position of its target (argument of the symbol).
\end{example}

\begin{remark}[On Adhesivity]
A category is \emph{adhesive}~\cite{lack2004adhesive} if (i)~it has all pullbacks, (ii)~it has all pushouts along monos, and (iii)~pushouts along monos are stable and are pullbacks~\cite[Theorem 3.2]{garner2012axioms}. Adhesivity implies the uniqueness of pushout complements (up to isomorphism)~\cite[Lemma 4.5]{lack2004adhesive}, which in turn ensures that DPO rewriting in adhesive categories is deterministic. Moreover, certain meta-properties of interest (such as the Local Church-Rosser Theorem and the Concurrency Theorem) hold for DPO rewriting in adhesive categories (and many graphical structures are adhesive). For these reasons, DPO and adhesivity are closely related in the literature.

We make two observations in connection to adhesivity. First, \pbpostrong rewriting makes strictly weaker assumptions than DPO rewriting: it is enough to assume conditions (i) and (ii) above. Second, \GraphLattice{(\labels,\leq)} is non-adhesive for any choice of complete lattice in which the maximum $\top$ and minimum $\bot$ are distinct: the square
\begin{center}
\begin{tikzcd}
\bot \arrow[d, tail] \arrow[r, tail] & \top \arrow[d, tail] \\
X \arrow[r, tail]                 & \top                
\end{tikzcd}
\end{center}
is a pushout along a mono for $X \in \{\bot,\top\}$, and hence admits two pushout complements. Thus, not only can \pbpostrong be applied to non-adhesive categories, \GraphLattice{(\labels,\leq)} is a graphical non-adhesive category with practical relevance. To what extent meta-properties of interest carry over to \pbpostrong/\GraphLattice{(\labels,\leq)} rewriting from DPO rewriting on adhesive categories is left for future work.
\end{remark}

\section{Discussion}
\label{sec:discussion}

We discuss our rewriting (Section~\ref{sec:discussion:rewriting}) and relabeling (Section~\ref{sec:discussion:relabeling}) contributions in turn.

\subsection{Rewriting}
\label{sec:discussion:rewriting}

Unlike other algebraic approaches such as DPO~\cite{ehrig1973graph}, SPO~\cite{lowe1993algebraic}, SqPO~\cite{corradini2006sesqui} and AGREE~\cite{corradini2015agree}, computing a rewrite step in PBPO and \pbpostrong requires only a basic understanding of constructing pullbacks and pushouts. Moreover, assuming monic matching, and under some mild restrictions (DPO is left-linear, and SPO uses conflict-free matches), Corradini et al.~\cite{corradini2006sesqui, corradini2019pbpo} have shown that
\[
\text{DPO} < \text{SPO} < \text{SqPO} < \text{AGREE} < \text{PBPO}
\]
where $\mathcal{F} < \mathcal{G}$ means that any $\mathcal{F}$ rule $\rho$ can be simulated by a $\mathcal{G}$ rule $\sigma$, and where \emph{simulation} means
${\Rightarrow_\rho^\mathcal{F}} \subseteq {\Rightarrow_\sigma^\mathcal{G}}$ for the generated rewrite relations. This chain may now be extended by inserting $\text{AGREE} < \text{\pbpostrong} < \text{PBPO}$.

\newcommand{\precc}{\stackrel{\star}{\prec}}
If instead of simulation one uses $\emph{modeling}$ as the expressiveness criterion, i.e., $\subseteq$ is strengthened to $=$, then the situation is different. Writing $\prec$ for this expressiveness relation, we conjecture that in $\Graph$ and with monic matching (which implies conflict-freeness of SPO matches)
\begin{center}
    \begin{tikzpicture}[default,nodes={rectangle,inner sep=3mm}]
        \node (t) {\pbpostrong};
        \node (l) at (t.west) [anchor=east,yshift=-3.5mm]  
          {$\text{SPO} \prec \text{SqPO} \prec \text{AGREE}$};
        \node (r) at (t.east) [anchor=west,xshift=-2mm,yshift=-3.5mm]  
          {$\text{DPO}$};
        \node (b) at ($(t)+(0,-7mm)$) {PBPO};
        \node at ($(l.north east)!0.5!(t.south west)$) [rotate=35] {$\prec$};
        \node at ($(r.north west)!0.5!(t.south east) + (-.5mm,0mm)$) [rotate=180-35] {$\prec$};
        \node at ($(t)!0.5!(b)$) [rotate=90] {$\prec$};
    \end{tikzpicture}
\end{center}
holds,\footnote{%
The modeling of DPO and SPO rules for category \Graph in \pbpostrong is similar to the approach described in our paper on PGR~\cite{overbeek2020patch}.}
and the other comparisons do not hold. $\text{PBPO} \prec \text{\pbpostrong}$ follows from Lemma~\ref{lemma:pbpo:monic:typing} and the fact that $\Graph$ is strongly amendable, and PBPO does not stand in any other modeling relation due to uncontrolled global effects (in particular, a straightforward adaptation of Example~\ref{ex:pbpo:remove:loop} shows $\text{DPO} \not\prec \text{PBPO}$).

Other graph rewriting approaches that bear certain similarities to \pbpostrong (see also the discussion in~\cite{corradini2019pbpo}) include the double-pullout graph rewriting approach by Kahl~\cite{kahl2010amalgamating}; variants of the aforementioned formalisms, such as the cospan SqPO approach by Mantz~\cite[Section 4.5]{mantz2014phd}; and the recent drag rewriting framework by  Dershowitz and Jouannaud~\cite{dershowitz2019drags}. Double-pullout graph rewriting also uses pullbacks and pushouts to delete and duplicate parts of the context (extending DPO), but the approach is defined in the context of collagories~\cite{kahl2011collagories}, and to us it is not yet clear in what way the two approaches relate. Cospan SqPO can be understood as being almost dual to SqPO: rules are cospans, and transformation steps consists of a pushout followed by a final pullback complement. An interesting question is whether \pbpostrong can also model cospan SqPO. Drag rewriting is a non-categorical approach to generalizing term rewriting, and like \pbpostrong, allows relatively fine control over the interface between pattern and context, thereby avoiding issues related to dangling pointers and the construction of pushout complements. Because drag rewriting is non-categorical and drags have inherently more structure than graphs, it is difficult to relate \pbpostrong and drag rewriting precisely. These could all be topics for future investigation.

Finally, let us just note that the combination of \pbpostrong and \GraphLattice{(\labels, \leq)} does not provide a strict generalization of Patch Graph Rewriting (PGR)~\cite{overbeek2020patch}, our conceptual precursor to \pbpostrong (Section~\ref{sec:introduction}). This is because patch edge endpoints that lie in the context graph can be redefined in PGR (e.g., the direction of edges between context and pattern can be inverted), but not in \pbpostrong. 
Beyond that, \pbpostrong is more general and expressive. Therefore, at this point we believe that the most distinguishing and redeeming feature of PGR is its visual syntax, which makes rewrite systems much easier to define and communicate. In order to combine the best of both worlds, our aim is to define a similar syntax for (a suitable restriction of) \pbpostrong in the future.

\subsection{Relabeling}
\label{sec:discussion:relabeling}

The coupling of \pbpostrong and \GraphLattice{(\labels, \leq)} allows relabeling and modeling sorts and variables with relative ease, and does not require a modification of the rewriting framework. Most existing approaches study these topics in the context of DPO, where the requirement to ensure the unique existence of a pushout complement requires restricting the method and proving non-trivial properties:
\begin{itemize}
    \item 
      Parisi-Presicce et al.~\cite{parisi1986graph} limit DPO rules $L \leftarrow K \to R$ to ones where $K \to R$ is monic (meaning merging is not possible), and where some set-theoretic consistency condition is satisfied. Moreover, the characterization of the existence of rewrite step has been shown to be incorrect~\cite{habel2002relabelling}, supporting our claim that pushout complements are not easy to reason about.
    \item 
      Habel and Plump~\cite{habel2002relabelling} study relabeling using the category of partially labeled graphs. They allow non-monic morphisms $K \to R$, but they nonetheless add two restrictions to the definition of a DPO rewrite rule. Among others, these conditions do not allow hard overwriting arbitrary labels as in Example~\ref{example:hard:overwriting}. Moreover, the pushouts of the DPO rewrite step must be restricted to pushouts that are also pullbacks.
      Finally, unlike the approach suggested by Parisi-Presice et al., Habel and Plump's approach does not support modeling notions of sorts and variables.
      
      Our conjecture that \pbpostrong can model DPO in $\Graph$ extends to this relabeling approach in the following sense: given a DPO rule over graphs partially labeled from $\labels$ that moreover satisfies the criteria of~\cite{habel2002relabelling}, we conjecture that there exists a \pbpostrong rule in \GraphLattice{(\labels^{\bot,\top}, \leq)} that models the same rewrite relation when restricting to graphs totally labeled over the base label set $\labels$.
\end{itemize}
Later publications largely appear to build on the approach~\cite{habel2002relabelling} by Habel and Plump. For example, Schneider~\cite{schneider2005changing} gives a non-trivial categorical formulation; Hoffman~\cite{hoffman2005graph} proposes a two-layered (set-theoretic) approach to support variables; and Habel and Plump~\cite{habel2012adhesive} generalize their approach to $\mathcal{M}, \mathcal{N}$-adhesive systems (again restricting $K \to R$ to monic arrows).

The transformation of attributed structures has been explored in a very general setting by Corradini et al.~\cite{corradini2019pbpo}, which involves a comma category construction and suitable restrictions of the PBPO notions of rewrite rule and rewrite step. We leave relating their and our approach to future work.

\subsubsection*{Acknowledgments} We thank Andrea Corradini and anonymous reviewers for useful discussions, suggestions and corrections.
We would also like to thank Michael Shulman, who identified the sufficient conditions for amendability for us~\cite{schulman2021mathoverflow}.
The authors received funding from the Netherlands Organization for Scientific Research (NWO) under the Innovational Research Incentives Scheme Vidi (project.\ No.\ VI.Vidi.192.004).

\bibliographystyle{unsrt}
\bibliography{main}

\begin{thebibliography}{10}

\bibitem{overbeek2020patch}
R.~Overbeek and J.~Endrullis.
\newblock Patch graph rewriting.
\newblock In {\em Proc.\ Conf.\ on Graph Transformation ({ICGT})}, volume 12150
  of {\em LNCS}, pages 128--145. Springer, 2020.

\bibitem{corradini2019pbpo}
A.~Corradini, D.~Duval, R.~Echahed, F.~Prost, and L.~Ribeiro.
\newblock The {PBPO} graph transformation approach.
\newblock {\em J.\ Log.\ Algebraic Methods Program.}, 103:213--231, 2019.

\bibitem{ehrig1973graph}
H.~Ehrig, M.~Pfender, and H.~J. Schneider.
\newblock Graph-grammars: An algebraic approach.
\newblock In {\em Proc.\ Symp.\ on on Switching and Automata Theory (SWAT)},
  page 167–180. IEEE Computer Society, 1973.

\bibitem{mac1971categories}
S.~Mac~Lane.
\newblock {\em Categories for the Working Mathematician}, volume~5.
\newblock Springer Science \& Business Media, 1971.

\bibitem{awodey2006category}
S.~Awodey.
\newblock {\em {Category Theory}}.
\newblock Oxford University Press, 2006.

\bibitem{ehrig2006}
H.~Ehrig, K.~Ehrig, U.~Prange, and G.~Taentzer.
\newblock {\em Fundamentals of Algebraic Graph Transformation}.
\newblock Springer, 2006.

\bibitem{habel2001doublerevisited}
A.~Habel, J.~M{\"{u}}ller, and D.~Plump.
\newblock Double-pushout graph transformation revisited.
\newblock {\em Math. Struct. Comput. Sci.}, 11(5):637--688, 2001.

\bibitem{corradini2015agree}
A.~Corradini, D.~Duval, R.~Echahed, F.~Prost, and L.~Ribeiro.
\newblock {AGREE} -- algebraic graph rewriting with controlled embedding.
\newblock In {\em Proc.\ Conf.\ on Graph Transformation (ICGT)}, volume 9151 of
  {\em LNCS}, pages 35--51. Springer, 2015.

\bibitem{corradini2019rewriting}
A.~Corradini, T.~Heindel, B.~K{\"{o}}nig, D.~Nolte, and A.~Rensink.
\newblock Rewriting abstract structures: Materialization explained
  categorically.
\newblock In {\em Proc.\ Conf.\ on Foundations of Software Science and
  Computation Structures (FOSSACS)}, volume 11425 of {\em LNCS}, pages
  169--188. Springer, 2019.

\bibitem{lack2004adhesive}
S.~Lack and P.~Soboci\'{n}ski.
\newblock Adhesive categories.
\newblock In {\em Proc.\ Conf.\ on Foundations of Software Science and
  Computation Structures (FOSSACS)}, volume 2987 of {\em LNCS}, pages 273--288.
  Springer, 2004.

\bibitem{garner2012axioms}
R.~Garner and S.~Lack.
\newblock On the axioms for adhesive and quasiadhesive categories.
\newblock {\em Theory and Applications of Categories}, 27(3):27--46, 2012.

\bibitem{lowe1993algebraic}
M.~L{\"{o}}we.
\newblock Algebraic approach to single-pushout graph transformation.
\newblock {\em Theor. Comput. Sci.}, 109(1{\&}2):181--224, 1993.

\bibitem{corradini2006sesqui}
A.~Corradini, T.~Heindel, F.~Hermann, and B~K{\"{o}}nig.
\newblock Sesqui-pushout rewriting.
\newblock In {\em Proc.\ Conf.\ on Graph Transformation (ICGT)}, volume 4178 of
  {\em LNCS}, pages 30--45. Springer, 2006.

\bibitem{kahl2010amalgamating}
W.~Kahl.
\newblock Amalgamating pushout and pullback graph transformation in
  collagories.
\newblock In {\em Proc.\ Conf.\ on Graph Transformation (ICGT)}, volume 6372 of
  {\em Lecture Notes in Computer Science}, pages 362--378. Springer, 2010.

\bibitem{mantz2014phd}
F.~Mantz.
\newblock {\em Coupled Transformations of Graph Structures applied to Model
  Migration}.
\newblock PhD thesis, University of Marburg, 2014.

\bibitem{dershowitz2019drags}
N.~Dershowitz and J.{-}P. Jouannaud.
\newblock Drags: {A} compositional algebraic framework for graph rewriting.
\newblock {\em Theor. Comput. Sci.}, 777:204--231, 2019.

\bibitem{kahl2011collagories}
W.~Kahl.
\newblock Collagories: Relation-algebraic reasoning for gluing constructions.
\newblock {\em J. Log. Algebraic Methods Program.}, 80(6):297--338, 2011.

\bibitem{parisi1986graph}
F.~Parisi{-}Presicce, H.~Ehrig, and U.~Montanari.
\newblock Graph rewriting with unification and composition.
\newblock In {\em Proc.\ Workshop on Graph-Grammars and Their Application to
  Computer Science}, volume 291 of {\em LNCS}, pages 496--514. Springer, 1986.

\bibitem{habel2002relabelling}
A.~Habel and D.~Plump.
\newblock Relabelling in graph transformation.
\newblock In {\em Proc.\ Conf.\ on Graph Transformation ({ICGT})}, volume 2505
  of {\em LNCS}, pages 135--147. Springer, 2002.

\bibitem{schneider2005changing}
H.~J. Schneider.
\newblock Changing labels in the double-pushout approach can be treated
  categorically.
\newblock In {\em Formal Methods in Software and Systems Modeling}, volume 3393
  of {\em LNCS}, pages 134--149. Springer, 2005.

\bibitem{hoffman2005graph}
B.~Hoffmann.
\newblock Graph transformation with variables.
\newblock In {\em Formal Methods in Software and Systems Modeling}, volume 3393
  of {\em LNCS}, pages 101--115. Springer, 2005.

\bibitem{habel2012adhesive}
A.~Habel and D.~Plump.
\newblock $\mathcal{M}$, $\mathcal{N}$-adhesive transformation systems.
\newblock In {\em Proc.\ Conf.\ on Graph Transformation ({ICGT})}, volume 7562
  of {\em LNCS}, pages 218--233. Springer, 2012.

\bibitem{schulman2021mathoverflow}
M.~Shulman.
\newblock Subobject- and factorization-preserving typings.
\newblock MathOverflow.
\newblock URL:https://mathoverflow.net/q/381933 (version: 2021-01-22).

\end{thebibliography}

\appendix
\section*{Appendix}
\label{sec:appendix}

\section{\pbpostrong{}}

We will need both directions of the well known pullback lemma.

\medskip
\noindent
\begin{minipage}{.74\textwidth}
    \begin{lemma}[Pullback Lemma]\label{lemma:pullback:lemma}
        Consider the diagram on the right.
        Suppose the right square is a pullback square and the left square commutes. Then the outer square is a pullback square iff the left square is a pullback square. \qed
    \end{lemma}
\end{minipage}\hfill
\begin{minipage}{.22\textwidth}
    \hfill
    \begin{tikzcd}[row sep=3mm, column sep=5mm]
    A \arrow[r] \arrow[d] & B \arrow[d] \arrow[r]                        & C \arrow[d] \\
    D \arrow[r]           & E \arrow[r] \arrow[ru, "\text{PB}", phantom] & F          
    \end{tikzcd}
\end{minipage}

\use{lem:topleft:pullback}

\begin{proof}
    In the following diagram, $u$ satisfying $t_K = u' \circ u$ and $m \circ l = g_L \circ u$ is inferred by using that $G_K$ is a pullback and commutation of the outer square.
    \begin{center}
    \begin{tikzcd}
    L \arrow[d, "m", tail] \arrow[dd, "t_L"', tail, bend right=49] & K \arrow[dd, "t_K", tail, bend left=49] \arrow[l, "l"'] \arrow[d, "u", dotted, tail] \\
    G_L \arrow[d, "\alpha"] \arrow[rd, "\text{PB}", phantom]       & G_K \arrow[l, "g_L"'] \arrow[d, "u'"]                                                \\
    L'                                                             & K' \arrow[l, "l'"']                                                                 
    \end{tikzcd}
    \end{center}
    By direction $\Longrightarrow$ of the pullback lemma (Lemma~\ref{lemma:pullback:lemma}), the created square is a pullback square, and so by stability of monos under pullbacks, $u$ is monic. \qed
\end{proof}

\use{lem:uniqueness:u}

\begin{proof}
    In the following diagram, the top-right pullback is obtained using Lemma~\ref{lem:topleft:pullback}, and the top-left pullback is a rotation of the match diagram:
    \begin{center}
      \begin{tikzpicture}[node distance=13mm,l/.style={inner sep=1mm},baseline=-7.5mm]
        \node (G) {$G_L$};
        \node (LP) [below of=G] {$L'$};
          \draw [->] (G) to node [left,l] {$\alpha$} (LP);
         \node (GKP) [right of=G] {$G_K$};
         \draw [->] (GKP) to node [above,l] {$g_L$} (G);
         \node (KP) [below of=GKP] {$K'$};
         \draw [->] (GKP) to node [right,l] {$u'$} (KP);
         \draw [->] (KP) to node [below,l] {$l'$} (LP);
         \node at ([shift={(-4mm,-4mm)}]GKP) {PB};
         \node (K) [above of=GKP] {$K$};
         \draw [>->] (K) to node [right,l] {$u$} (GKP);
         \node (L) [above of=G] {$L$};
         \draw [->] (K) to node [above,l] {$l$} (L);
         \draw [>->] (L) to node [right,l] {$m$} (G);
         \node at ([shift={(-4mm,-4mm)}]K) {PB};
         \node (L') [left of=G] {$L'$};
         \node (L2) [left of=L] {$L$};
         \draw [>->] (L) to node [above,l] {$1_L$} (L2);
         \draw [>->] (L2) to node [left,l] {$t_L$} (L');
         \draw [->] (G) to node [below,l] {$\alpha$} (L');
         \node at ([shift={(-4mm,-4mm)}]L) {PB};
         \node (K2) [above right of=K] {$K$};
         \draw [->, bend right=35] (K2) to node [above,l] {$1_L \circ l$} (L2);
         \draw [->, bend left=35] (K2) to node [right,l] {$v$} (GKP);
         \draw [->, dotted] (K2) to node [above left,l] {$!x$} (K);
      \end{tikzpicture}
    \end{center}
    By direction $\Longleftarrow$ of the pullback lemma, $L \xleftarrow{1_L \circ l} K \xrightarrow{u} G_K$ is a pullback for the topmost outer square.
    
    Now suppose that for a morphism $v : K \to G_K$, $t_K = u' \circ v$. Then
    $
    \alpha \circ g_L \circ u = l' \circ u' \circ u = l' \circ t_K = l' \circ u' \circ v = \alpha \circ g_L \circ v$.
    Hence both $v$ and $u$ make the topmost outer square commute. Hence there exists a unique $x$ such that (simplifying) $l \circ x = l$ and $u \circ x = v$. From known equalities and monicity of $t_K$ we then derive
    \[
    \begin{array}{rcl}
    u \circ x & = & v \\
    u' \circ u \circ x & = & u' \circ v \\
    t_K \circ x & =  & t_K \\
    t_K \circ x & = & t_K \circ 1_K \\
    x & = & 1_K\text{.}
    \end{array}
    \]
    Hence $u = v$. \qed
\end{proof}

\use{lem:bottomright:pushout}

\begin{proof}
    The argument is similar to the proof of Lemma~\ref{lem:topleft:pullback}, but now uses the dual statement of the pullback lemma. \qed
\end{proof}

\section{Expressiveness of \pbpostrong{}}

\use{lem:compacted:rule:lemma}
\begin{proof}
By using the following commuting diagram:
\begin{center}
  \begin{tikzpicture}[node distance=12mm,l/.style={inner sep=.5mm},baseline=-6mm]
    \node (L) {$L$};
    \node (K) [right of=L] {$K$}; \draw [->] (K) to node [above] {$l$} (L);
    \node (Lc) [below of=L] {$L_c$};
      \draw [->>] (L) to node [right] {$e$} (Lc);
    \node (Kc) [below of=K] {$K_c$};
      \draw [->] (K) to node [right] {} (Kc);
      \draw [->] (Kc) to node [below] {} (Lc);
       \node at ([shift={(-4mm,-4mm)}]K) {PB};
    \node (R) [right of=K] {$R$}; 
      \draw [->] (K) to node [above] {$r$} (R);
    \node (Rc) [below of=R] {$R_c$};
      \draw [->] (R) to node [right] {} (Rc);
      \draw [->] (Kc) to node [below] {} (Rc);
      \node at ([shift={(-4mm,4mm)}]Rc) {PO};
      
    \node (GL) [below of=Lc] {$G_L$};
    \node (GK) [below of=Kc] {$G_K$};
    \node (GR) [below of=Rc] {$G_R$};
    \node at ([shift={(-4mm,-4mm)}]GK) {PB};
    \node at ([shift={(-4mm,4mm)}]GR) {PO};
      
    \node (LP) [below of=GL] {$L'$};
      \draw [->] (GL) to node [left] {$\alpha$} (LP);
    \node (KP) [below of=GK] {$K'$};
      \draw [->] (GK) to node [right] {} (KP);
      \draw [->] (KP) to node [below] {$l'$} (LP);
       \node at ([shift={(-4mm,-4mm)}]Kc) {PB};
    \node (RP) [below of=GR] {$R'$};
      \draw [->] (KP) to node [below] {$r'$} (RP);
      \draw [->] (GR) to node [right] {} (RP);
      \node at ([shift={(-4mm,4mm)}]RP) {PO};
     
    \draw [>->] (Kc) to (GK);
    \draw [->] (Rc) to (GR);
    \draw [->] (GK) to (GL);
    \draw [->] (GK) to (GR);
    
    \draw [->, bend right=70] (L) to node [left] {$t_L$} (LP);
    \draw [>->] (Lc) to node [right] {$m'$} (GL);
    \draw [->, bend right=60] (Lc) to node [right] {$t_{L_c}$} (LP);
  \end{tikzpicture}
\end{center}
\qed
\end{proof}

\use{thm:pbpo:plus:models:pbpo}
\begin{proof}
From Corollary~\ref{corollary:pbpo:monic:matching} we obtain a set of PBPO rules $S$ that collectively model $\rho$ using monic matching, and by Lemma~\ref{lemma:pbpo:monic:typing} each $\sigma \in S$ can be modeled by a monic rule $\tau_\sigma$ with a strong matching rewrite policy. By Proposition~\ref{prop:basic:facts:pbpo:restrictions}, the set $\{ \tau_\sigma \mid \sigma \in S \}$ corresponds to a set of \pbpostrong rules. \qed
\end{proof}

\section{Category \GraphLattice{(\labels,\leq)}}

\use{lem:graph:lattice:strongly:amendable}
{
\renewcommand{\src}{\mathrm{src}}
\renewcommand{\tgt}{\mathrm{tgt}}

\newcommand{\forget}[1]{[{#1}]}

\newcommand{\gleq}{\GraphLattice{(\labels,\leq)}}

\newcommand{\lift}[3]{{#1}^{({#2}, {#3})}}

\newcommand{\incl}[1]{{#1}_\top}

\newcommand{\UGraph}{\catname{UGraph}}

\begin{proof}

Let \UGraph refer to the category of unlabeled graphs.

Given a $t_L : L \to L'$ in \gleq, momentarily forget about the labels and consider the unlabeled version (overloading names) in \UGraph. Because \UGraph is a topos, it is strongly amendable. Thus we can obtain a factorization $L \stackrel{t_L'}{\mono} L'' \stackrel{\beta}{\to} L'$ on the level of \UGraph that witnesses strong amendability, i.e., for any factorization of $L \stackrel{m}{\mono} G \stackrel{\alpha}{\to} L'$ of $t_L$ in \UGraph, there exists an $\alpha'$ such that
    \begin{center}
        \begin{tikzcd}
            L \ar[rr, bend left, "t_L"] \ar[r, tail, "m"] \ar[d, tail, "1_L"] & G \ar[r, "\alpha"] \ar[d, "\alpha'"'] & L' \\
            L \ar[r, tail, "t_L'"'] & L'' \ar[ru, "\beta"'] &
        \end{tikzcd}
    \end{center}
commutes and the left square is a pullback square.

The idea now is to lift the bottom unlabeled factorization into \gleq. As far as graph structure is concerned, we know that it is a suitable factorization candidate. Then all that needs to be verified are the order requirements $\leq$ on the labels.

For the lifting of $L''$, choose the graph in which every element has the same label as its image under $\beta$. Then clearly the lifting of $\beta$ of \UGraph into \gleq is well-defined, and so is the lifting of $t_L'$ (using that $t_L$ is well-defined in \gleq).

Now given any factorization $L \stackrel{m}{\mono} G \stackrel{\alpha}{\to} L'$ in \gleq, lift the $\alpha'$ that is obtained by considering the factorization on the level of \UGraph. Then the lifting of $\alpha'$ is well-defined by $\alpha = {\beta \circ \alpha'}$ and well-definedness of $\alpha$ and $\beta$ in \gleq. All that remains to be checked is that the left square is a pullback as far as the labels are concerned, i.e., whether for every $x \in V_L \cup E_L$, $\lbl_L(x) = {\lbl_L(1_L(x)) \land \lbl_G(m(x))}$. This follows using $\lbl_L(x) \leq \lbl_G(m(x))$ and the complete lattice law $\forall a \ b \ldotp a \leq b \implies a \land b = a$. \qed
\end{proof}

}

\end{document}